\documentclass[conference]{IEEEtran}
\makeatletter
\def\ps@headings{%
\def\@oddhead{\mbox{}\scriptsize\rightmark \hfil \thepage}%
\def\@evenhead{\scriptsize\thepage \hfil \leftmark\mbox{}}%
\def\@oddfoot{}%
\def\@evenfoot{}}
\makeatother
\usepackage{graphicx}
\usepackage{amssymb}
\usepackage{epstopdf}
\usepackage{amsmath}
\usepackage{enumerate}
\usepackage{cite}
\usepackage{verbatim}
\usepackage{mathcomp}
\usepackage{supertabular}
\usepackage{longtable}
\usepackage{stmaryrd}
\usepackage{url}
\usepackage[usenames,dvipsnames]{pstricks}
\usepackage{epsfig}
\usepackage{pst-grad}
\usepackage{pst-plot}
\usepackage[singlelinecheck=false, font = small, justification = raggedright, skip = -35pt]{caption}[2008/04/01]
\newcommand{\C}{{\mathcal C}}
\newcommand{\mc}{{\mathcal{C}}}
\newcommand{\cl}{{\mathcal{C}(\boldsymbol{L})}}
\newcommand{\clw}{{\mathcal{C}(\boldsymbol{L}[W])}}

\newcommand{\E}{{\mathcal E}}

\newcommand{\G}{{\mathcal G}}

\newcommand{\I}{{\mathcal I}}
\newcommand{\J}{{\mathcal J}}

\newcommand{\X}{{\mathcal X}}
\newcommand{\Xah}{{\widehat{\mathcal X}_A}}
\newcommand{\Xa}{{\mathcal{X}_A}}
\newcommand{\Z}{{\mathcal Z}}

\newcommand{\Y}{{\mathcal Y}}
\newcommand{\V}{{\mathcal V}}

\newcommand{\N}{{\mathcal N}}
\newcommand{\FXZ}{{\mathcal {F}(m,n,\X,\Z,f)}}

\newcommand{\vr}{{\sf r}}

\newcommand{\bM}{{\boldsymbol M}} 
\newcommand{\bI}{{\boldsymbol I}} 
\newcommand{\bP}{{\boldsymbol P}} 
\newcommand{\bQ}{{\boldsymbol Q}}

\newcommand{\bL}{{\boldsymbol L}}
\newcommand{\bLo}{{\boldsymbol L}^{(0)}}
\newcommand{\bA}{{\boldsymbol A}}

\newcommand{\bX}{\boldsymbol{X}}
\newcommand{\bY}{\boldsymbol{Y}}

\newcommand{\sa}{{\mathsf{a}}}

\newcommand{\en}{{\sf H}}

\newcommand{\mi}{{\sf I}}
\newcommand{\supp}{{\sf supp}}
\newcommand{\sfG}{{\mathsf{G}}}
\newcommand{\dist}{{\mathsf{d}}}
\newcommand{\weight}{{\mathsf{wt}}}
\newcommand{\spn}{{\mathsf{span}_q}}

\newcommand{\rank}{{\mathsf{rank}_q}}
\newcommand{\mr}{{\text{min-rank}}}

\newcommand{\define}{\stackrel{\mbox{\tiny $\triangle$}}{=}}

\newcommand{\ba}{{\boldsymbol a}}
\newcommand{\bb}{{\boldsymbol b}}
\newcommand{\bu}{{\boldsymbol u}}
\newcommand{\bv}{{\boldsymbol v}}
\newcommand{\bs}{{\boldsymbol s}}
\newcommand{\by}{{\boldsymbol y}}
\newcommand{\byi}{{{\boldsymbol y}^{(i)}}}

\newcommand{\bc}{{\boldsymbol c}}
\newcommand{\bch}{{\hat{{\boldsymbol c}}}}
\newcommand{\bO}{{\boldsymbol 0}}
\newcommand{\bw}{{\boldsymbol{w}}}
\newcommand{\bg}{{\boldsymbol{g}}}

\newcommand{\bx}{{\boldsymbol{x}}}
\newcommand{\bz}{{\boldsymbol{z}}}

\newcommand{\bG}{{\boldsymbol{G}}}
\newcommand{\be}{{\boldsymbol e}}
\newcommand{\bei}{{{\boldsymbol e}^{(i)}}}

\newcommand{\bbt}{\boldsymbol{\beta}}

\newcommand{\bcone}{{\boldsymbol c}^{(1)}}
\newcommand{\bctwo}{{\boldsymbol c}^{(2)}}
\newcommand{\bck}{{\boldsymbol c}^{(k)}}
\newcommand{\bci}{{\boldsymbol c}^{(i)}}

\newcommand{\bvi}{{\boldsymbol v}^{(i)}}
\newcommand{\bui}{{{\boldsymbol u}^{(i)}}}

\newcommand{\NX}{{\N_q(\delta,m,n,\X,f)}}

\newcommand{\IX}{{\I(q, m, n, \X, f)}}
\newcommand{\JX}{{\J(m, n, \X, f)}}
\newcommand{\bxi}{{\boldsymbol{\xi}}}
\newcommand{\bxii}{{\boldsymbol \xi^{(i)}}}

\newcommand{\bbN}{{\mathbb N}}

\newcommand{\al}{\alpha}
\newcommand{\bt}{\beta}
\newcommand{\vep}{\varepsilon}

\newcommand{\kp}{\kappa}
\newcommand{\kpq}{\kappa_q}

\newcommand{\dd}{{$\delta$-error-correcting $(m,n,\X,f)$-IC over $\fq$}}
\newcommand{\mic}{{$(m,n,\X,f)$-IC over $\fq$}}
\newcommand{\rmic}{{$\eta$-randomized $(m,n,\X,f)$-IC over $\fq$}}

\newcommand{\ra}{\rightarrow}
\renewcommand{\ge}{\geq}
\renewcommand{\le}{\leq}

\newcommand{\ff}{\mathbb{F}}
\newcommand{\fq}{\mathbb{F}_q}

\newcommand{\fqn}{\mathbb{F}_q^n}
\newcommand{\et}{{\emph{et al.}}}

\newcommand{\fkD}{{\mathfrak D}}
\newcommand{\fkE}{{\mathfrak E}}

\newtheorem{definition}{Definition}[section]
\newtheorem{example}{Example}[section]
\newtheorem{theorem}{Theorem}[section]
\newtheorem{proposition}[theorem]{Proposition}
\newtheorem{lemma}[theorem]{Lemma}

\newtheorem{corollary}[theorem]{Corollary}
\newtheorem{remark}[theorem]{Remark}

\begin{document}

\title{On the Security of Index Coding with Side Information}

\author{
  \IEEEauthorblockN{Son Hoang Dau$^*$, Vitaly Skachek$^{\dagger,1}$, and Yeow Meng Chee$^\ddagger$ \vspace{1ex}}
  \IEEEauthorblockA{$^{*,\ddagger}$Division of Mathematical Sciences,
    School of Physical and Mathematical Sciences\\
    Nanyang Technological University,
    21 Nanyang Link, Singapore 637371 \vspace{1ex}\\
    $^{\dagger}$Coordinated Science Laboratory, University of Illinois at Urbana-Champaign \\
    1308 W. Main Street, Urbana, IL 61801, USA \vspace{1ex} \\
    Emails: {\it $^{*}$daus0002@ntu.edu.sg, $^{\dagger}$vitalys@illinois.edu, $^{\ddagger}$YMChee@ntu.edu.sg }
  }
}

\maketitle

\footnotetext[1]{The work of this author was done while he was with the Division of Mathematical Sciences, School of Physical and Mathematical Sciences, Nanyang Technological University, 21 Nanyang Link, Singapore 637371.}

\begin{abstract}
\boldmath 
Security aspects of the Index Coding with Side Information (ICSI) problem are investigated. 
Building on the results of Bar-Yossef \emph{et al.} (2006), the properties of linear index codes are further explored. 
The notion of weak security, considered by Bhattad and Narayanan (2005) in the context of network coding, is 
generalized to \emph{block security}. 
It is shown that the linear index code based on a matrix $\bL$, 
whose column space code $\cl$ has length $n$, minimum distance $d$ and dual distance $d^\perp$, is $(d-1-t)$-block secure (and hence also weakly secure) if the adversary knows in advance $t \leq d-2$ messages, and is completely insecure if the adversary knows in advance more than $n - d^\perp$ messages. 
Strong security is examined under the conditions that the adversary: (i) possesses $t$ messages in advance; 
(ii) eavesdrops at most $\mu$ transmissions; (iii) corrupts at most $\delta$ transmissions. 
We prove that for sufficiently large $q$, 
an optimal linear index code which is strongly secure against such an adversary has length $\kp_q + \mu+2\delta$. Here
$\kp_q$ is a generalization of the {\mr} over $\fq$ of the side information graph for the ICSI problem in its original formulation in the work of Bar-Yossef~\emph{et~al.} 
\end{abstract}

\section{Introduction}
\label{sec:introduction}

\PARstart{T}he problem of Index Coding with Side Information (ICSI) was introduced by Birk and Kol~\cite{BirkKol98}, \cite{BirkKol2006}. It was motivated by applications such as audio and video-on-demand, and daily newspaper delivery. In these applications a server (sender) has to deliver some sets of data, audio or video files to a set of clients (receivers), 
different sets are requested by different receivers. Assume that before the transmission starts, 
the receivers have already (from previous transmissions) some files or movies in their possession.
Via a slow backward channel, the receivers can let the sender know which messages they already have in their possession, and which messages they request. By exploiting this information, the amount of the overall transmissions can be reduced. 
As it was observed in~\cite{BirkKol98}, this can be achieved by coding the messages at the server 
before broadcasting them out. 

Another possible application of the ICSI problem is in opportunistic wireless networks. 
These are networks in which a wireless node can opportunistically listen to the wireless channel. As a result, the node may obtain packets that were not designated to it (see~\cite{Rouayheb2009, Katti2006, Katti2008}). This way, a node obtains some side information about the transmitted data. Exploiting this additional knowledge may help to increase the throughput of the system. 

Consider the toy example in Figure~1. It presents a scenario with one sender and four receivers. 
Each receiver requires a different information packet (or message). 
The na\"ive approach requires four separate transmissions, one transmission per an information
packet. However, by exploiting the knowledge of the subsets of messages that clients already have, and by using coding of the transmitted data, the server can satisfy all the demands by broadcasting just one coded packet. 

The ICSI problem has been a subject of several recent studies \cite{Yossef, LubetzkyStav, Wu, Rouayheb2007, Rouayheb2008, Rouayheb2009, ChaudhrySprintson, Alon}. This problem can be regarded as a special case of the well-known network coding (NC) problem~\cite{Ahlswede}, \cite{KoetterMedard2003}. In particular, it was shown that every instance of the NC problem can be reduced to an instance of the ICSI problem~\cite{Rouayheb2008, Rouayheb2009}.
\vspace{-45pt}
\begin{figure}[h]
\begin{flushright}
\scalebox{1} 
{
\begin{pspicture}(0,-4.9003124)(7.7978125,5.0596876)
\definecolor{color1106b}{rgb}{0.8,0.8,0.8}
\pscircle[linewidth=0.04,dimen=outer,fillstyle=solid,fillcolor=color1106b](3.3278124,0.0496875){0.55}
\pscircle[linewidth=0.04,dimen=outer,fillstyle=solid,fillcolor=color1106b](1.3678125,1.9896874){0.51}
\pscircle[linewidth=0.04,dimen=outer,fillstyle=solid,fillcolor=color1106b](1.3278126,-1.9703125){0.51}
\pscircle[linewidth=0.04,dimen=outer,fillstyle=solid,fillcolor=color1106b](5.3478127,2.0096874){0.51}
\pscircle[linewidth=0.04,dimen=outer,fillstyle=solid,fillcolor=color1106b](5.3278127,-1.9303125){0.51}
\usefont{T1}{ptm}{m}{n}
\rput(3.3552186,0.0296875){$S$}
\usefont{T1}{ptm}{m}{n}
\rput(5.315219,-1.9303125){$R_3$}
\usefont{T1}{ptm}{m}{n}
\rput(5.315219,2.0096874){$R_4$}
\usefont{T1}{ptm}{m}{n}
\rput(1.3552188,1.9896874){$R_1$}
\usefont{T1}{ptm}{m}{n}
\rput(1.3152188,-1.9703125){$R_2$}
\pscircle[linewidth=0.022,linestyle=dashed,dash=0.16cm 0.16cm,dimen=outer](3.3578124,0.0196875){1.04}
\pscircle[linewidth=0.018,linestyle=dashed,dash=0.16cm 0.16cm,dimen=outer](3.3478124,0.0496875){1.75}
\pscircle[linewidth=0.02,linestyle=dashed,dash=0.16cm 0.16cm,dimen=outer](3.3378124,0.0396875){2.56}
\usefont{T1}{ptm}{m}{n}
\rput(4.8192186,0.0696875){$\sum_{i = 1}^4 x_i$}
\usefont{T1}{ptm}{m}{n}
\rput(5.777656,3.2296875){has $x_1,x_2,x_3$}
\usefont{T1}{ptm}{m}{n}
\rput(5.6076565,2.7896874){requests $x_4$}
\usefont{T1}{ptm}{m}{n}
\rput(1.4176563,3.2496874){has $x_2,x_3,x_4$}
\usefont{T1}{ptm}{m}{n}
\rput(1.2276562,2.7896874){requests $x_1$}
\usefont{T1}{ptm}{m}{n}
\rput(1.3776562,-2.7503126){has $x_1,x_3,x_4$}
\usefont{T1}{ptm}{m}{n}
\rput(1.2276562,-3.1903124){requests $x_2$}
\usefont{T1}{ptm}{m}{n}
\rput(5.777656,-2.7503126){has $x_1,x_2,x_4$}
\usefont{T1}{ptm}{m}{n}
\rput(5.6076565,-3.2103126){requests $x_3$}
\end{pspicture} 
}
\end{flushright}
\caption{An example of the ICSI problem}
\end{figure}
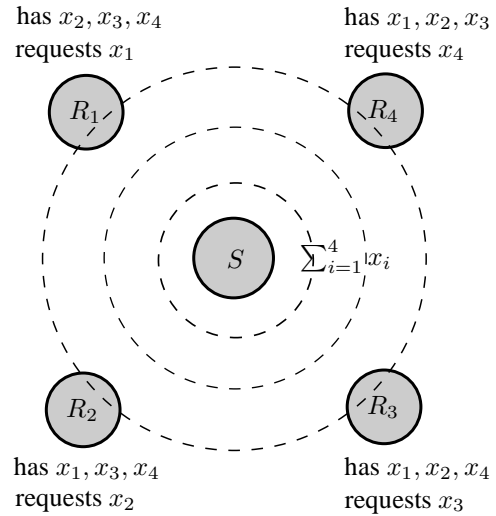

Several previous works focused on the design of an efficient index code for the ICSI problem. Given an instance of the ICSI problem, Bar-Yossef \emph{et al.} \cite{Yossef} proved that finding the best \emph{scalar linear binary index code} is equivalent to finding the so-called \emph{\mr} of a graph, which is known to be an NP-hard problem (see \cite{Yossef, Peeters96}). Here scalar linear index codes refer to linear index codes in which each message is a symbol in the field $\fq$. By contrast, in \emph{vector linear index codes} each message is a vector over $\fq$. Lubetzky and Stav \cite{LubetzkyStav} showed that there exist instances in which scalar linear index codes over nonbinary fields and linear index codes over mixed fields outperform the scalar linear binary index codes. 
El Rouayheb \et \cite{Rouayheb2008, Rouayheb2009} showed that for certain instances of the ICSI problem, vector linear index codes achieve strictly higher transmission rate than scalar linear index codes do. They also pointed out that there exist instances in which vector nonlinear index codes outperform vector linear index codes. Vector nonlinear index codes were also shown to outperform scalar nonlinear index codes 
for certain instances by Alon {\et} \cite{Alon}. 
Several heuristic solutions for the ICSI problem were proposed in \cite{Rouayheb2007, ChaudhrySprintson}. 
          
In this paper, we study the security aspects of linear index codes. We restrict ourselves to \emph{scalar linear} index codes. 
It is known that vector linear index codes can achieve better transmission rate than their scalar counterparts, for certain instances of the ICSI problem~\cite{Rouayheb2008, Rouayheb2009}. However, if the block length is fixed, one can model a 
vector index code as a scalar index code applied to another instance of the ICSI problem. 
If the block length is $\ell$, the number of messages is $n$, and the number of receivers is $m$
in the original (vector) instance, then the equivalent (scalar) instance can be viewed as having $\ell n$ messages
and $\ell m$ receivers. 

Let $\fq$ be a finite field with $q$ elements. 
In its most general formulation, a linear index code maps $\bx \in \fqn$ onto $( \bx | \bg) \, \bL$, 
where $\bg \in \fq^\eta$ is a random vector, $\bL$ is an $(n + \eta) \times N$ matrix over $\fq$, and $n, \eta, N \in \bbN$. 
In this work, we show that each deterministic linear index code (i.e. $\eta = 0$) provides a certain level of 
information security. 
More specifically, let the code $\cl$ be spanned by the columns of $\bL$, and let
$d$ and $d^\perp$ be its minimum distance and dual distance, respectively. 
We say that a particular adversary is of strength $t$ if it has $t$ messages in its possession. 
Then, we show that the index code based on $\bL$ is $(d-1-t)$-block secure against all adversaries 
of strength $t \leq d - 2$ and is completely insecure against any adversary of strength at least 
$n - d^\perp + 1$. If $\cl$ is an MDS code, then the two bounds coincide. The technique used in the 
proof for this result is reminiscent of that used in the constructions of (multiple) secret sharing schemes
from linear error-correcting codes~\cite{Massey1993,Ding1997}. The results on the security of linear index codes
can be further employed to analyze the existence of solutions for a natural generalization of the ICSI problem, 
so-called the \emph{Index Coding with Side and Restricted Information (ICSRI)} problem. 
In that problem, it is required that some receivers have no information about some messages.     

In the sequel, we also consider a non-deterministic linear index code, based on the use of random symbols (i.e. $\eta \ge 1$). We show that the coset coding technique (which has been successfully employed in Secure Network Coding literature, see, for instance \cite{CaiYeung2002, Feldman2004, Rouayheb_Soljanin2007, Zhuang2010, Silva_Kschischang2010}) gives an optimal strongly secure linear index code of length $\kp_q + \mu + 2\delta$. This index code is strongly secure against an adversary which: 
\begin{itemize}
\item[(i)] has $t$ arbitrary messages in advance; 
\item[(ii)] eavesdrops at most $\mu$ transmissions; 
\item[(iii)] corrupts at most $\delta$ transmissions. 
\end{itemize}

Previous works on the security aspects (and on the error-correction aspect, as a special case) 
of network coding dealt with the multicast scenario. One of the main reasons for this limitation is that the optimal simultaneous transmission rates for non-multicast networks have not been fully characterized yet. The ICSI problem
can be modeled as a special case of the non-multicast Network Coding problem (\cite{Alon, Rouayheb2009}). 
Moreover, being modeled in that way, 
it requires that there are directed edges from particular sources to each sink, which provide the side information. 
The symbols transmitted on these special edges are not allowed to be corrupted, where  usually for network coding any edge can be corrupted. These two differences restrict the ability to derive the results on the security 
of the index coding schemes from the existing results on security of network coding schemes.  

The paper is organized as follows. Notations and definitions, which are used in the rest of the paper, are introduced in Section~\ref{sec:preliminaries}. The model and some basic results for the ICSI problem are presented in Section \ref{sec:icsi}. The block security of linear index codes is analyzed in Section~\ref{sec:block_security}. In particular, the Index Coding with Side and Restricted Information problem is presented and analyzed in Section~\ref{sec:ICSRI}. 
Section~\ref{sec:strong_security} is devoted to the analysis of strong security for index coding. 
The paper is concluded in Section~\ref{sec:conclusion}.  

\section{Preliminaries}
\label{sec:preliminaries}
Recall that we use the notation $\fq$ for the finite field with $q$ elements, where $q$ is a power of prime. We also use $\fq^*$ for the set of all nonzero elements of $\fq$. Let $[n]$ denote the set of integers $\{1,2,\ldots,n\}$.
For the vectors $\bu = (u_1, u_2, \ldots, u_n) \in \fq^n$ and $\bv = (v_1, v_2, \ldots, v_n) \in \fq^n$, 
the (Hamming) distance between $\bu$ and $\bv$ is defined to be the number of coordinates where $\bu$ and $\bv$ differ, 
namely, 
\[
\dist(\bu,\bv) = |\{i \in [n] \; : \; u_i \ne v_i\}|. 
\]
The \emph{support} of a vector $\bu \in \fqn$ is defined to be the set $\text{supp}(\bu) = \{i \in [n]: u_i \ne 0\}$.  
The (Hamming) weight of a vector $\bu$, denoted $\weight(\bu)$, is defined to be $|\text{supp}(\bu)|$, the number of nonzero coordinates of $\bu$. 

A $k$-dimensional subspace $\mc$ of $\fq^n$ is called a linear $[n,k,d]_q$ ($q$-ary) code if the minimum distance of $\mc$, 
\[
\dist(\mc) \define \min_{\bu \in \mc, \; \bv \in \mc, \; \bu \neq \bv} \dist(\bu,\bv) \; ,
\]
is equal to $d$. Sometimes we may use the notation $[n,k]_q$ for the sake of simplicity. The vectors in $\mc$ are called codewords. It is easy to see that the minimum weight of a nonzero codeword in a linear code $\mc$ is equal to its minimum distance $\dist(\mc)$. A \emph{generator matrix} $\sfG$ of an $[n,k]_q$-code $\mc$ is a $k \times n$ matrix whose rows are linearly independent codewords of $\mc$. Then $\mc = \{\by \sfG : \by \in \fq^k\}$. 

The \emph{dual code} or \emph{dual space} of $\mc$ is defined as $\mc^\bot = \{\bu \in \fq^n: \; \bu \bc^T = 0 \text{ for all } \bc \in \mc\}$. The minimum distance of $\mc^\bot$, $\dist(\mc^\bot)$, is called the dual distance of $\mc$.     

The following upper bound on the minimum distance of a $q$-ary linear code is well-known (see~\cite{MW_and_S} Chapter 1).

\medskip
\begin{theorem}[Singleton bound]
\label{singleton}
For an $[n,k,d]_q$-code, we have $d \leq n - k + 1$. 
\end{theorem}
\medskip 

Codes attaining this bound are called \emph{maximum distan\-ce se\-pa\-rable} (MDS) codes. For a subset of vectors 
$$
\{ \bcone, \bctwo,\ldots,\bck \} \subseteq \fq^n \; , 
$$ define its linear span over $\fq$:  
\begin{multline*}
\spn\left( \{ \bcone,\bctwo, \ldots, \bck \} \right) \define \\ \left\{ \sum_{i = 1}^k \al_i \bci \; : 
\; \al_i \in \fq, \; i \in [k] \right\} \; .
\end{multline*}
We use $\be_i = (\underbrace{0,\ldots,0}_{i-1},1,\underbrace{0,\ldots,0}_{n-i}) \in \fqn$ to denote the unit vector, which has a one at the $i$th position, and zeros elsewhere. We also use $\mathbf{I}_n$, 
$n \in \bbN$, to denote the $n \times n$ identity matrix.   

We recall the following well-known result in coding theory.

\medskip
\begin{theorem}[\cite{Hedayat}, p. 66]
\label{thmOA}
Let $\mc$ be an $[n,k,d]_q$-code with dual distance $d^\perp$ and $\bM$ denote the $q^k \times n$ matrix whose $q^k$ rows are codewords of $\mc$. If $r \leq d^\perp-1$ then each $r$-tuple from $\fq$ appears in an arbitrary set of $r$ columns of $\bM$ exactly $q^{k-r}$ times. 
\end{theorem}
\medskip

For a random vector $\bY = (Y_1,Y_2,\ldots,Y_n)$ and a subset $B = \{i_1,i_2,\ldots,i_b\}$ of $[n]$, where $i_1 < i_2 < \cdots <i_b$, let $\bY_B$ denote the vector $(Y_{i_1},Y_{i_2},\ldots,Y_{i_b})$. For an $n \times k$ matrix $\bM$, let $\bM_i$ denote the $i$th row of $\bM$, and $\bM[j]$ its $j$th column. For a set $E \subseteq [n]$, let $\bM_E$ denote the $|E| \times k$ submatrix of $\bM$ formed by rows of $\bM$ which are indexed by the elements of $E$. 
For a set $F \subseteq [k]$, let $\bM[F]$ denote the $n \times |F|$ submatrix of $\bM$ 
formed by columns of $\bM$ which are indexed by the elements of $F$. 

Let $X$ and $Y$ be discrete random variables taking values in the sets $\Sigma_X$ and $\Sigma_Y$, respectively. Let $\text{Pr}(X=x)$ denote the probability that $X$ takes a particular value $x \in \Sigma_X$. Let $\en(X)$, $\en(X|Y)$, $\mi(X;Y)$, and $\mi(X;Y|Z)$ denote the (binary) entropy, conditional entropy, mutual information, and conditional mutual information (see \cite{CoverThomas1991} for the background).  

\section{Index Coding and Some Basic Results}
\label{sec:icsi}

The Index Coding with Side Information (ICSI) problem considers the following communications scenario. 
There is a unique sender (or source) $S$, who has a vector of messages $\bx = (x_1, x_2, \ldots, x_n) \in \fqn$ in his possession, which is a realized value of a random vector $\bX = (X_1,X_2,\ldots,X_n)$. $X_1,X_2,\ldots,X_n$ hereafter are assumed to be independent uniformly distributed 
random variables over $\fq$. 
There are also $m$ receivers $R_1,R_2,\ldots,R_m$. For each $i \in [m]$, $R_i$ has some side information, 
i.e. $R_i$ owns a subset of messages $\{ x_j \}_{j \in \X_i}$, $\X_i \subseteq [n]$. In addition, each $R_i$, $i \in [m]$, 
is interested in receiving the message $x_{f(i)}$, for some \emph{demand function} $f: [m] \rightarrow [n]$. 
Here we assume that $f(i) \notin \X_i$ for all $i \in [m]$. Let $\X=(\X_1,\X_2,\ldots,\X_m)$. 
An instance of the ICSI problem is given by a quadruple $(m,n,\X,f)$. 
Here we assume that every receiver requests exactly one message. This assumption is not a limitation of the model, as we can consider an equivalent problem by splitting each receiver who requests multiple messages into multiple receivers,
each of whom requests exactly one message and have the same set of side information (see \cite{BirkKol98,Yossef}). 

\vskip 10pt 
\begin{definition} 
An \emph{index code} over $\fq$ for an instance $(m, n,\X,f)$ of the ICSI problem, 
referred to as an {\mic}, is an encoding function
\begin{eqnarray*}
\fkE & : & \fq^n \rightarrow \fq^N \; , 
\end{eqnarray*}
such that for each receiver $R_i$, $i \in [m]$, there exists a decoding function
\[
\fkD_i \: : \: \fq^N \times \fq^{|\X_i|} \rightarrow \fq \; , \\
\]
satisfying
\[
\forall \bx \in \fq^n \; : \; \fkD_i(\fkE(\bx), \bx_{\X_i}) = x_{f(i)} \; .
\]
The parameter $N$ is called the \emph{length} of the index code. 
In the scheme corresponding to this code, 
$S$ broadcasts a vector $\fkE(\bx)$ of length $N$ over $\fq$.   
\label{def:IC}
\end{definition}

\vskip 10pt
\begin{definition}
An index code of the shortest possible length is called \emph{optimal}.
\end{definition} 

\medskip
\begin{definition} 
A \emph{linear index code} is an index code, for which the encoding function 
$\fkE$ is a linear transformation over $\fq$. 
Such a code can be described as 
\[
\forall \bx \in \fq^n \; : \; \fkE (\bx) = \bx \bL \; , 
\]
where $\bL$ is an $n \times N$ matrix over $\fq$. The matrix $\bL$ is called the \emph{matrix corresponding 
to the index code $\fkE$}. We also refer to $\fkE$ as the \emph{index code based on $\bL$}.  
Notice that the length of $\fkE$ is the number of columns of $\bL$. 
\end{definition}
\vskip 10pt 

Let $E \subseteq [n]$ and $\bu \in \fqn$. 
In the sequel, we write $\bu \lhd E$ if $\supp(\bu) \subseteq E$. 
Intuitively, this means that if some receiver knows $x_j$ for all $j \in E$ (and also knows $\bu$), then 
this receiver is also able to compute the value of $\bx \bu^T$.

Hereafter, we assume that the sets $\X_i$, for all $i \in [m]$, are known to $S$.
Moreover, we also assume that the index code $\fkE$ is known to each receiver $R_i$, $i \in [m]$. In practice 
this can be achieved by a preliminary communication session, when the knowledge of the sets 
$\X_i$, for all $i \in [m]$, and of the code $\fkE$ are disseminated between the participants of 
the scheme. 

Let $\C(\bL) = \spn(\{ \bL[j]^T \}_{j \in [N]})$, the subspace spanned by the (transposed) columns of $\bL$. 
The following lemma was implicitly formulated in \cite{Yossef} for the case where $m = n$, $f(i) =i$ for 
all $i \in [m]$, and $q = 2$. 
This lemma specifies a sufficient condition on $\C(\bL)$ so that a receiver can reconstruct a particular message. 
We reproduce this lemma with its proof in its general form for the sake of completeness of the presentation.  
 
\vskip 10pt 
\begin{lemma}
\label{lem1}
Let $\bL$ be an $n \times N$ matrix over $\fq$. Assume that $S$ broadcasts $\bx \bL$. 
Then, for each $i \in [m]$, the receiver $R_i$ can reconstruct $x_{f(i)}$ if 
there exists a vector $\bui \in \fqn$ satisfying 
\begin{enumerate}
\item $\bui \lhd \X_i$; 
\item $\bui + \be_{f(i)} \in \cl$. 
\end{enumerate}
\end{lemma}
\vskip 10pt

\begin{proof}
Assume that $\bui \lhd \X_i$ and $\bui + \be_{f(i)} \in \cl$. 
Since $\bui + \be_{f(i)} \in \cl$, 
there exists $\boldsymbol{\bt} \in \fq^N$ such that 
\[
\bui + \be_{f(i)} = \bbt \bL^T. 
\]
By taking the transpose and pre-multiplying by $\bx$, we obtain 
\[
\bx(\bui + \be_{f(i)})^T = (\bx \bL) \bbt^T.
\]
Therefore, 
\[
x_{f(i)} = \bx \be_{f(i)}^T = (\bx \bL) \bbt^T - \bx \bui^T.
\]
Observe that $R_i$ is able to find $\bui$ and $\bbt$ from the knowledge of $\bL$.
Moreover, $R_i$ is also able to compute $\bx \bui^T$ since $\bui \lhd \X_i$. Additionally, $R_i$ knows 
$\bx \bL$, which is transmitted by $S$.
Therefore, $R_i$ is able to compute $x_{f(i)}$. 
\end{proof} 
\medskip

\begin{remark}
It follows from Lemma~\ref{lem1} that $\bL$ corresponds to a linear {\mic} if  $\cl \supseteq \spn( \{\bui + \be_{f(i)} \}_{i \in [m]})$, for some $\bui \lhd \X_i$, $i \in [m]$. We show later in Corollary \ref{coro:index_code} that this condition is also necessary. Finding such an $\bL$ with minimal number of columns by careful selection of $\bui$'s is a difficult task (in fact it is NP-hard to do so, see \cite{Yossef, Peeters96}), which, however, yields a linear coding scheme with the minimal number of transmissions. 
\end{remark}

\section{Block Secure Linear Index Codes}
\label{sec:block_security}
\subsection{Block Security and Weak Security}

In this section, we assume the presence of an adversary $A$ who can listen to all transmissions. 
Assume that $S$ employs a linear index code based on $\bL$. 
The adversary is assumed to possess side information $\{x_j \}_{j \in \X_A}$, 
where $\X_A \subsetneq [n]$. 
For short, we say that $A$ knows (or possesses, owns) $\bx_{\X_A}$. 
The \emph{strength} of $A$ is defined to be $|\X_A|$. 
Denote $\widehat{\X}_A \define \left( [n] \backslash \X_A \right)$. 
Note that by listening to $S$, the adversary also knows $\bs \define \fkE(\bx) = \bx \bL$. We define below several levels 
of security for linear index codes. 

\vskip 10pt 
\begin{definition}
\label{defSecurity}
Suppose that the sender $S$ possesses a vector of messages $\bx \in \fq^n$, which is a realized value of a random vector 
$\bX = (X_1, X_2, \ldots, X_n)$, whose coordinates $X_i$, $i \in [n]$, are all independent and uniformly distributed over $\fq$. An adversary $A$ possesses $\bx_{\X_A}$. Consider a linear {\mic} based on $\bL$. 
\begin{enumerate}
\item 
For $B \subseteq \widehat{\X}_A$, the adversary is said to \emph{have no information about $\bx_B$} if 
\begin{equation}
\label{defnoinfo}
\en(\bX_B| \bX \bL,\bX_{\X_A}) = \en(\bX_B). 
\end{equation}
In other words, despite the partial knowledge on $\bx$ that the adversary has (his side information and the transmissions he eavesdrops), the symbols $\bx_B$ still looks completely random to him.  
\item The index code is said to be \emph{$b$-block secure against $\X_A$} if for every $b$-subset $B \subseteq \widehat{\X}_A$, the adversary has no information about $\bx_B$.  
\item 
The index code is said to be \emph{$b$-block secure against all adversaries of strength $t$} ($0 \le t \le n-1$) if it is $b$-block secure against $\X_A$ for every $\X_A \subset [n]$, $|\X_A| = t$. 
\item  
The index code is said to be \emph{weakly secure against $\X_A$} if it is $1$-block secure against $\X_A$. In other words, after listening to all transmissions, the adversary has no information about each particular message that he does not possess in the first place. 
\item
The index code is said to be \emph{weakly secure against all adversaries of strength $t$ } ($0 \le t \le n-1$) if it is weakly secure against $\X_A$ for every $t$-subset $\X_A$ of $[n]$. 
\item
The index code is said to be \emph{completely insecure against $\X_A$} if an adversary, who possesses $\{x_i \}_{i \in \X_A}$, by listening to all transmissions, is able to determine $x_i$ for all $i \in \widehat{\X}_A$.
\item 
The index code is said to be \emph{completely insecure against any adversary of strength $t$} ($0 \le t \le n-1$) if an adversary, who possesses an arbitrary set of $t$ messages, is always able to reconstruct all of the other $n-t$ messages after listening to all transmissions.
\end{enumerate}
\end{definition}
\medskip 

\begin{remark}
Even when the index code is $b$-block secure ($b \geq 1$) as defined above, the adversary is still able to obtain information about dependencies between various $x_i$'s in $\widehat{\X}_A$ (but he gains no information about any group of $b$ particular messages). This definition of $b$-block security is a generalization of that of weak security (see~\cite{Bhattad},~\cite{Silva}). Obviously, if an index code is $b$-block secure against $\X_A$ ($b\geq 1$) then it is also weakly secure against $\X_A$, but the converse is not always true. 
\end{remark}

\subsection{Necessary and Sufficient Conditions for Block Security} 

In the sequel, we consider the sets $B \subseteq [n]$, $B \neq \varnothing$, and $E \subseteq [n]$, $E \neq \varnothing$. Moreover, we assume that the sets $\X_A$, $B$, and $E$ are disjoint, and that they form a partition of $[n]$, namely $\X_A \cup B \cup E = [n]$. In particular, $\widehat{\X}_A = B \cup E$. 

\vskip 10pt 
\begin{lemma}
\label{lem2}
Assume that for all $\bu \lhd \X_A$ and for all $\al_i \in \fq$, $i \in B$ (not all $\al_i$'s are zeros), 
\begin{equation}
\label{E1}
\bu + \sum_{i \in B}\al_i\be_{i} \notin \cl. 
\end{equation}
Then,
\begin{enumerate} 
\item for all $i \in B$:
\begin{equation}
\bL_i \in \spn(\{\bL_j\}_{j \in E}); 
\label{eq:span}
\end{equation}
\item 
the system
\begin{equation}
\label{E2}
\by \bL_E = \bw \bL_B  
\end{equation}
has at least one solution $\by \in \fq^{|E|}$ for every choice of $\bw \in \fq^{|B|}$. 
\end{enumerate}
\end{lemma}
\medskip
\begin{proof}
\begin{enumerate}
\item
If $\rank(\bL_E) = N$ then the first claim follows immediately. Otherwise, assume that $\rank(\bL_E) < N$. 
As the $N$ columns of $\bL_E$ are linearly dependent, there exists $\by \in \fq^N\backslash \{\bO\}$ such that $\by  \bL_E^T = \bO$. 
\begin{itemize}
\item
If for all such $\by$ and for all $i \in B$ we have $\by \bL_i^T = 0$, then $\bL_i \in ((\spn(\{\bL_j\}_{j \in E}))^\perp)^\perp = \spn(\{\bL_j\}_{j \in E})$ for all $i \in B$. 
\item
Otherwise, there exist $\by \in \fq^N$ and $i \in B$ such that $\by \bL_E^T = \bO$ and $\by \bL_i^T \neq 0$. Without loss of generality, assume that 
\[
\bL=\left[
\begin{array}{c}
\bL_{\X_A} \\ \hline
\bL_B \\ \hline
\bL_E
\end{array}\right].
\]
Let $\bc = \by \bL^T \in \cl$. Then
\[
\bc = (\bc_{\X_A}|\bc_B|\bc_E) = \big(\by \bL_{\X_A}^T \big| \by \bL_B^T \big| \by \bL_E^T \big). 
\]
Hence $\bc_B = \by \bL_B^T \neq \bO$ and $\bc_E = \by \bL_E^T = \bO$. Let $\bu = (\bc_{\X_A}|\bO|\bO) \lhd \X_A$ and $\al_i = c_i$ for all $i \in B$. Then $\al_i$'s are not all zero and $\bu + \sum_{i \in B}\al_i \be_i = \bc \in \cl$, which contradicts (\ref{E1}). 
\end{itemize}
\item
By~(\ref{eq:span}), each row of $\bL_B$ is a linear combination of rows of $\bL_E$. Hence $\bw \bL_B$ is also a linear combination of rows of $\bL_E$. Therefore, (\ref{E2}) has at least one solution.
\end{enumerate}
\end{proof}
\medskip

While Lemma~\ref{lem1} does not discuss security, it provides sufficient conditions 
for successful reconstruction of the information by a legitimate receiver $R_i$. 
Obviously, the legitimate receiver $R_i$ can be replaced by an adversary $A$ in the formulation 
of that lemma. Thus, the conditions in Lemma~\ref{lem1} can be viewed as sufficient conditions 
for absence of weak security with respect to this $A$. 

In Lemma~\ref{lem3}, which appears below, we show that the same conditions are also necessary (for the absence 
of information security with respect to the adversary $A$). However, similarly, $A$ can be replaced by 
a legitimate receiver $R_i$ in the formulation of Lemma~\ref{lem3}. Thus, this lemma also provides both necessary 
and sufficient conditions for successful reconstruction of the information by the legitimate receiver $R_i$. 

Additionally, in Lemma~\ref{lem3}, the weak security is further generalized to \emph{block security}. 

\vskip 10pt 
\begin{lemma}
\label{lem3} 
Let $\bL$ be an $n \times N$ matrix over $\fq$. Assume that $S$ broadcasts $\bx \bL$.  
For a subset $B \subseteq \widehat{\X}_A$, an adversary $A$ who owns $\bx_{\X_A}$, 
after listening to all transmissions, has no information about $\bx_B$ if and only if
\begin{equation}
\label{E3}
\begin{split}
\forall \bu \lhd \X_A,\ &\forall \al_i \in \fq \text{ with } \al_i,i\in B, \text{ not all zero:}\\ &\bu + \sum_{i \in B}\al_i\be_{i} \notin \cl. 
\end{split}
\end{equation}
In particular, for each $i \in \widehat{\X}_A$, $A$ has no information about $x_i$ if and only if 
\[
\forall \bu \lhd \X_A \; : \bu + \be_i \notin \cl. 
\]
\end{lemma}
\medskip

\begin{proof}
Assume that~(\ref{E3}) holds. 
We need to show that $\en(\bX_B| \bX \bL, \bX_{\X_A}) = \en(\bX_B)$.
It suffices to show that for all $\bg \in \fq^{|B|}$:
\begin{equation}
\label{E4}
\text{Pr}(\bX_B = \bg| \bX \bL = \bs,\ \bX_{\X_A} = \bx_{\X_A})= \dfrac{1}{q^{|B|}}, 
\end{equation}
where $\bs = \bx \bL$ for some $\bx \in \fqn$.

Consider the following linear system with the unknown $\bz \in \fq^n$
\begin{equation*}
\begin{cases} \bz_B = \bg \\ \bz_{\X_A} = \bx_{\X_A} \\ \bz \bL = \bs \end{cases}  , 
\end{equation*}
which is equivalent to 
\begin{equation}
\label{E5}
\begin{cases} \bz_B = \bg\\ \bz_{\X_A} = \bx_{\X_A} \\ \bz_E \bL_E = \bs {-} \bg \bL_B {-} \bx_{\X_A} \bL_{\X_A}   \end{cases}. 
\end{equation}

In order to prove that (\ref{E4}) holds, it suffices to show that for all choices of $\bg \in \fq^{|B|}$, (\ref{E5}) always has the same number of solutions $\bz$. Notice that the number of solutions $\bz$ of (\ref{E5}) is equal to the number of solutions $\bz_E$ of 
\begin{equation}
\label{E6}
\bz_E \bL_E  = \bs - \bg \bL_B  - \bx_{\X_A} \bL_{\X_A}, 
\end{equation}
where $\bs$, $\bg$, and $\bx_{\X_A}$ are known. For any $\bg \in \fq^{|B|}$, if (\ref{E6}) has a solution, then it has exactly $q^{|E|-\rank(\bL_E)}$ different solutions. Therefore, it suffices to prove that (\ref{E6}) has at least one solution for every $\bg \in \fq^{|B|}$. 

Since $\bx$ is an obvious solution of (\ref{E5}), we have
\begin{equation}
\label{E7}
\bx_E \bL_E  = \bs - \bx_B \bL_B  - \bx_{\X_A}  \bL_{\X_A}. 
\end{equation}
Subtract (\ref{E7}) from (\ref{E6}) we obtain
\[
(\bz_E - \bx_E) \bL_E = (\bx_B - \bg) \bL_B,
\]
which can be rewritten as
\begin{equation}
\label{E8}
\by \bL_E = \bw \bL_B,
\end{equation}
where $\by \define \bz_E - \bx_E$, $\bw \define \bx_B - \bg$. Due to Lemma \ref{lem2}, (\ref{E8}) always has a solution $\by$, for every choice of $\bw$. Therefore (\ref{E6}) has at least one solution for every $\bg \in \fq^{|B|}$. 

Now we prove the converse. Assume that (\ref{E3}) does not hold. Then there exist $\bu \lhd \X_A$ and $\al_i \in \fq$, $i \in B$, where $\al_i$'s, $i \in B$ are not all zero, such that
\[
\sum_{i \in B} \al_i \be_i = \bc - \bu \; ,
\] 
for some $\bc \in \cl$. Hence, similar to the proof of Lemma \ref{lem1}, the adversary obtains
\begin{eqnarray*}
\sum_{i \in B} \al_i x_i & = & \bx \left( \sum_{i \in B} \al_i \be_i \right)^T \\
& = & \bx (\bc - \bu)^T \\
& = & \bx \bc^T - \bx \bu^T. 
\end{eqnarray*}
Note that the adversary can calculate $\bx \bc^T$ from $\bs$, and can also find $\bx \bu^T$ based on his own side information. Therefore, $A$ is able to compute a nontrivial linear combination of $x_i$'s, $i \in B$. Hence the entropy $\en(\bX_B| \bX \bL, \bX_{\X_A}) < \en(\bX_B)$. Thus, the adversary gains some information about the $\bx_B$. 
\end{proof}
\vskip 10pt 

Corollary~\ref{coro1} generalizes Lemma~\ref{lem1} by providing both necessary and sufficient conditions 
for a receiver's ability to recover the desired message. (Note that this corollary can be equally 
applied to the legitimate receiver $R_i$ as well as to the adversary $A$.)

\vskip 10pt 
\begin{corollary}
\label{coro1}
Let $\bL$ be an $n \times N$ matrix over $\fq$ and let $S$ broadcast $\bx \bL$. 
Then for each $i \in [m]$, the receiver $R_i$ can reconstruct $x_{f(i)}$ if and only if there exists a vector $\bui \in \fqn$ such that 
\begin{enumerate}
\item $\bui \lhd \X_i$; 
\item $\bui + \be_{f(i)} \in \cl$.
\end{enumerate}
\end{corollary}
\vskip 10pt 

\begin{corollary}
\label{coro:index_code}
The matrix $\bL$ corresponds to a linear {\mic} if and only if for all $i \in [m]$, there exists 
a vector $\bui \in \fq^n$ satisfying
\begin{enumerate}
	\item $\bui \lhd \X_i$; 
  \item $\bui + \be_{f(i)} \in \cl$.
\end{enumerate}
\end{corollary}
\vskip 10pt 

\begin{remark}
It follows from Corollary~\ref{coro:index_code} that $\bL$ corresponds to a linear {\mic} if and only if $\cl \supseteq \spn( \{ \bui + \be_{f(i)} \}_{i \in [m]})$, for some $\bui \lhd \X_i$, $i \in [m]$. If we define
\begin{equation} 
\label{equ:E111}
\begin{split}
\kpq &= \kpq(m,n,\X,f) \\
&\define \min\{\rank(\{\bui + \be_{f(i)}\}_{i \in [m]}): \bui \in \fqn, \bui \lhd \X_i\},
\end{split} 
\end{equation} 
then $\kpq$ is the shortest possible length of a linear {\mic}. 
\end{remark}
\vskip 10pt 

\begin{corollary}
\label{coro:minrank}
The length of an optimal linear {\mic} is $\kpq = \kpq(m,n,\X,f)$. 
\end{corollary}
\vskip 10pt 

\begin{remark}
The quantity $\kpq$ defined in (\ref{equ:E111}) is precisely the $\mr$ over $\fq$ of the side information graph
of an ICSI instance in the case $m = n$ and $f(i) = i$ for all $i \in [n]$. 
\begin{proof}[Idea of the proof] 
Recall that in \cite{Yossef}, the \emph{side information graph} $\G$ of an instance of the ICSI problem is defined by 
$\G=(\V_\G, \E_\G)$, where $\V_\G = [n]$ and 
\[
\E_\G = \{e=(i,j):\ i,j \in [n],\ j \in \X_i\}. 
\]
A matrix $\bA$ over $\fq$ is said to \emph{fit} $\G$ (\cite{Haemers1978}) if  
\[
\begin{cases} a_{i,j} \neq 0,& \text{ if } i = j, \\ a_{i,j} = 0,& \text{ if } i \neq j, \ (i,j) \notin \E_\G. \end{cases}
\]
Then the \emph{$\mr$} of the side information graph $\G$ is defined by
\[
\min \{\rank (\bA): \ \bA \text{ fits } \G \}. 
\] 
Suppose that $\bui \lhd \X_i$ for all $i \in [n]$. 
Let $\bA=(a_{i,j})$ be the $n \times n$ matrix whose $i$th row is precisely
$\bui + \be_i$, for each $i \in [n]$. Then $\bA$ fits $\G$.  
Conversely, if $\bA'$ fits $\G$ then by multiplying each row of $\bA'$ with a suitable nonzero 
constant (which does not change the rank of $\bA'$), one obtains a matrix $\bA$ which is of the aforementioned form. 
In other words, for each $i \in [n]$, the $i$th row of the resulting matrix $\bA$ equals $\bui + \be_i$ for some $\bui \lhd \X_i$. 
Therefore, $\kpq$ defined above is indeed the minimum rank over $\fq$ of a matrix which fits  
 the side information graph $\G$. Thus, $\kpq$ is precisely the $\mr$ over $\fq$ of $\G$. 
\end{proof}
\end{remark}
\vskip 10pt 

\begin{theorem}
\label{mainthm1}
Consider a linear {\mic} based on $\bL$. 
Let $d$ be the minimum distance of $\cl$. 
\begin{enumerate}
\item This index code is $(d-1-t)$-block secure against all adversaries of strength $t \le d-2$.
In particular, it is weakly secure against all adversaries of strength $t = d-2$.
\item This index code is not weakly secure against at least one adversary of strength $t = d-1$. Generally, 
if there exists a codeword of $\cl$ of weight $w$, then this index code is not weakly secure against at least one adversary of strength $t = w - 1$. 
\item Every adversary of strength $t \leq d - 1$ is able to determine a list of $q^{n-t-N}$ vectors in $\fq^n$ which includes the vector of messages $\bx$.   
\end{enumerate}
\end{theorem}
\medskip 

\begin{proof}
\begin{enumerate}
\item
Assume that $t \leq d - 2$. By Lemma \ref{lem3}, it suffices to show that for every $t$-subset $\X_A$ of $[n]$ and for every $(d-1-t)$-subset $B$ of $\widehat{\X}_A$,  
\begin{equation*}
\begin{split}
\forall \bu \lhd \X_A,\ &\forall \al_i \in \fq \text{ with } \al_i,i\in B, \text{ not all zero}:\\ &\bu + \sum_{i \in B}\al_i\be_{i} \notin \cl. 
\end{split}
\end{equation*}
For such $\bu$ and $\al_i$'s, we have $\weight(\bu + \sum_{i \in B}\al_i\be_{i}) = \weight(\bu) + \weight(\sum_{i \in B}\al_i\be_{i}) \leq t + (d-1-t) = d -1 < d$. Moreover, as $\supp(\bu) \cap B = \varnothing$ and $\al_i$'s, $i \in B$, are not all zero, we deduce that $\bu + \sum_{i \in B}\al_i\be_{i} \neq \bO$. We conclude that $\bu + \sum_{i \in B}\al_i\be_{i} \notin \cl$.

\item
We now show that the index code is not weakly secure against at least one adversary of strength $t = d-1$. The more general statement can be proved in an analogous way. 

Pick a codeword 
$\bc = (c_1, c_2, \ldots, c_n) \in \cl$ such that $\weight(\bc)=d$ and let $\supp(\bc) = \{i_1,i_2,\ldots,i_d\}$. Take 
$\X_A = \{i_1,i_2,\ldots,i_{d-1}\}$, $|\X_A| = d-1$. Let 
$$
\bu = (\bc / c_{i_d} - \be_{i_d}) \; .
$$
Then, $\bu \lhd \X_A$ and $\bu + \be_{i_d} = \bc / c_{i_d} \in \cl$. By Lemma \ref{lem1}, $A$ is able to determine $x_{i_d}$. Hence  
the index code is not weakly secure against the adversary $A$, who knows $d-1$ messages $x_i$'s in advance.

\item
Let $\bs = \bx \bL$. Consider the following linear system of equations with unknown $\bz \in \fq^n$
\begin{equation*}
\hspace{-15ex} \begin{cases}
\bz_{\X_A} = \bx_{\X_A} \\ 
\bz \bL  = \bs
\end{cases} 
, 
\end{equation*}
which is equivalent to
\begin{equation}
\begin{cases}
\bz_{\X_A} = \bx_{\X_A}\\ 
\bz_{\widehat{\X}_A} \bL_{\widehat{\X}_A} = \bs - \bx_{\X_A} \bL_{\X_A}
\end{cases}
\label{E10} . 
\end{equation}
The adversary $A$ attempts to solve this system. 
Given that $\bs$ and $\bx_{\X_A}$ are known, the system~(\ref{E10}) has $n - t$ unknowns and $N$ equations. Note that $t \leq d -1$, and thus by applying Theorem~\ref{singleton} to $\cl$ we have $n - t \geq n - d + 1 \geq N$. If $\rank(\bL_{\widehat{\X}_A}) = N$ then (\ref{E10}) has exactly $q^{n-t-N}$ solutions, as required. 

Next, we show that $\rank(\bL_{\widehat{\X}_A}) = N$. 
Assume, by contrary, that the $N$ columns of $\bL_{\widehat{\X}_A}$, 
denoted by $\bcone, \bctwo, \ldots, \bc^{(N)}$, are linearly dependent. 
Then there exist $\bt_i \in \fq$, $i \in [N]$, not all zero, such that $\sum_{i=1}^N \bt_i \bci = \bO$.
Let 
\[
\bc = \sum_{i=1}^N \bt_i \bL[i] \in \cl \backslash \{ \bO \} \; . 
\]
(Recall that $\bL[i]$ denotes the $i$th column of $\bL$). Then $\bc_{\widehat{\X}_A}=\sum_{i=1}^N \bt_i \bci = \bO$ and hence $\weight(\bc) = \weight(\bc_{\X_A}) \le t \leq d - 1$. This is a contradiction, which follows from the assumption that 
the $N$ rows of  $\bL_{\widehat{\X}_A}$ are linearly dependent.  
\end{enumerate}
\end{proof}
\vskip 10pt 

\begin{example}
Let $q = 2$. Assume that $\X_A = \varnothing$ and that $\X_i \ne \varnothing$ for all $i \in [m]$. 
For each $i \in [m]$ choose some $j_i \in \X_i$. 
Let $\bL$ be the binary matrix whose columns form a basis of the space $\cl = \spn(\{\be_{j_i} + \be_{f(i)}\}_{i \in [m]})$. 
Then $\dist(\cl) =2$. Since $t = |\X_A|=0$, we have $d-1-t = 1$. Therefore by 
Theorem~\ref{mainthm1} the index code based on $\bL$ is weakly secure against $A$. 
Moreover, if $\cl$ is nontrivial then $\bL$ has $N \leq n - d = n - 2$ columns. 
In other words, in that case, the index code based on $\bL$ requires at most $n - 2$ transmissions. 
\end{example}

\subsection{Block Security and Complete Insecurity}

Theorem~\ref{mainthm1} provides a threshold for the security level of a linear index code based on $\bL$. 
If $A$ has a prior knowledge of any $t \le d - 2$ messages, where $d = \dist(\cl)$, then the scheme is still secure, i.e. the adversary has no information about any group of $d-1-t$ particular messages from $\{ x_j \}_{j \in \widehat{\X}_A}$. 
On the other hand, the scheme may no longer be secure against an adversary of strength $t = d-1$. The last assertion of Theorem~\ref{mainthm1} shows us the difference between being block secure and being strongly secure. More specifically, if the scheme is strongly secure, the messages $\bx_\Xah$, which are not leaked to the adversary in advance, look completely random to the adversary, i.e. the probability to guess them correctly is $1/q^{n-t}$. However, if the scheme is $(d-1-t)$-block secure (for $t \leq d-2$), then the adversary is able to guess these messages correctly with probability $1/q^{n - t - N}$.  

For an adversary of strength $t \ge d$, the security of the scheme depends on the properties of the code employed, in particular, it depends on the weight distribution of $\cl$. From Theorem~\ref{mainthm1}, if there exists $\bc \in \cl$ with $\weight(\bc) = w$, then the scheme is not weakly secure against some adversary of strength $t=w - 1$. In general, the index code might still be ($b$-block or weakly) secure against some adversaries of strength $t$ for $t \geq d$. While we cannot make a general conclusion on the security of the scheme when the adversary's strength is larger than $d-1$, Lemma \ref{lem3} is still a useful tool to evaluate the security in that situation. However, as the next theorem shows, if the size of $\X_A$
is sufficiently large, then $A$ is able to determine all the messages in $\{ x_j \}_{j \in \widehat{\X}_A}$.   

\medskip
\begin{theorem}
\label{mainthm2}
The linear index code based on $\bL$ is completely insecure against any adversary of strength $t \ge n - d^\perp + 1$, where $d^\perp$ denotes the dual distance of $\cl$.  
\end{theorem}
\medskip 

\begin{proof}
Suppose the adversary knows a subset 
$\{x_j\}_{j \in \X_A}$, $\X_A \subsetneq [n]$ and $|\X_A| = t \geq n - d^\perp + 1$. By Corollary~\ref{coro1}, it suffices to show that for all $j \in \hat{\X}_A$, there exists $\bu \in \fqn$ satisfying simultaneously $\bu \lhd \X_A$ and $\bu + \be_j \in \cl$. 

Indeed, take any $j \in \widehat{\X}_A$, and let $\rho = n - t \le d^\perp -1$. 
Consider the $\rho$ indices which are not in $\X_A$. By Theorem~\ref{thmOA}, there exists a codeword $\bc \in \cl$ with 
\begin{equation*}
c_\ell = \begin{cases} 1 &\mbox{ if } \ell = j, \\ 0 & \mbox{ if } \ell \notin \X_A \cup \{ j \} 
\end{cases} \; . 
\end{equation*}
Then $\text{supp}(\bc) \subseteq \X_A \cup \{ j \}$.
We define $\bu \in \fqn$ such that $\bu \lhd \X_A$, as follows. For $\ell \in \X_A$, we set $u_\ell = c_\ell$, 
and for $\ell \notin \X_A$, we set $u_\ell = 0$. 
It is immediately clear that $\bc = \bu + \be_j$. 
Therefore, by Corollary~\ref{coro1}, the adversary can reconstruct $x_j$. 
We have shown that the index code is completely insecure against an arbitrary set $\X_A$ satisfying $|\X_A| \geq n - d^\perp + 1$, hence completing the proof. 
\end{proof}
\medskip

When $\cl$ is an MDS code, we have $n - d^\perp + 1 = d - 1$, and hence the two bounds established 
in Theorems~\ref{mainthm1} and~\ref{mainthm2} are actually tight. 
The following example further illustrates the results stated in these theorems. 
\medskip
\begin{example}
\label{ex1}
Let $n=7$, $m =7$, $q = 2$, and $f(i) = i$ for all $i \in [m]$. Suppose that the receivers have in their possession sets of messages as appear in the third column of the table below. Suppose also, that the demands of all receivers are as in the second column of the table. 
 
\begin{equation*}
\begin{array}{|c||c|c|}
\hline
\text{Receiver} & \text{Demand} & \{x_j\}_{i \in \X_i} \\
\hline
\hline
R_1 & x_1 & \{x_6, x_7\} \\
\hline
R_2 & x_2 & \{x_5, x_7\} \\
\hline
R_3 & x_3 & \{x_5, x_6\} \\
\hline
R_4 & x_4 & \{x_5, x_6, x_7 \} \\
\hline
R_5 & x_5 & \{x_1, x_2, x_6\} \\
\hline
R_6 & x_6 & \{x_1, x_3, x_4\} \\
\hline
R_7 & x_7 & \{x_2, x_3, x_6\} \\
\hline
\end{array}
\end{equation*}

For $i \in [7]$, let $\bui \in \ff_2^7$ such that $\text{supp}(\bui)=\X_i$.
Assume that an index code based on $\bL$ with $\cl = \spn(\{\bui + \be_i\}_{i \in [7]})$ is used. 
For instance, we can take $\bL$ to be the matrix whose set of columns is $\{\bL[i] \define \bui + \be_i\}_{i \in [4] }$. 
It is easy to see that $\cl$ is a $[7,4,3]_2$ Hamming code with $d=3$ and $d^\perp = 4$. 

Following the coding scheme, $S$ broadcasts the following four bits: 

\begin{center}
$s_1 = \bx (\bu^{(1)} + \be_1)^T$, \\
$s_2 =\bx (\bu^{(2)} + \be_2)^T$, \\
$s_3 =\bx (\bu^{(3)}+\be_3)^T$, \\
$s_4 =\bx (\bu^{(4)}+\be_4)^T$. \\
\end{center} 

Each $R_i$, $i \in [7]$, can compute $\bx (\bui + \be_i)^T$ by using a linear 
combination of $s_1,s_2,s_3,s_4$. Then, each $R_i$ can
subtract $\bx \bui^T$ (his side information) from $\bx (\bui + \be_i)^T$ to retrieve $x_i = \bx \be_i^T$. 

For example, consider $R_5$. Since 
\[
\bx \left(\bu^{(5)}+\be_5\right)^T = \bx \left( (\bu^{(1)}+\be_1)+(\bu^{(2)}+\be_2) \right)^T = s_1 + s_2, 
\]
$R_5$ subtracts $x_1 + x_2 + x_6$ from $s_1 + s_2$ to obtain 
\begin{eqnarray*}
&& \hspace{-6ex} (s_1 + s_2) - (x_1 + x_2 + x_6) \\
& = & (x_1 + x_2 + x_5 + x_6) - (x_1 + x_2 + x_6) \\
& = & x_5. 
\end{eqnarray*}

If an adversary $A$ has a knowledge of a single message $x_i$, then by Theorem~\ref{mainthm1}, $A$ is not able to determine any other message $x_\ell$, for $\ell \neq i$. Indeed, $\dist(\cl) = 3$, while $t = 1$, the code is weakly secure against all adversaries of strength $t =1$. If none of the messages are leaked, then the adversary has no information about any group of $2$ messages. On the other hand, the code is completely insecure against any adversary of strength $t \ge 4$; in that case $A$ is able to determine the remaining $7-t$ messages.    
\end{example}
\vskip 10pt 

\begin{remark}
So far we only discuss the case when the adversary can listen to all $N$ transmissions. 
If we consider an adversary, which can eavesdrop at most $\mu$ ($\mu \leq N$)
messages, then analogous results can also be obtained. 
Consider a linear index code based on $\bL$. Let 
\[
d_\mu \define \min \big\{\dist(\C(\bL[W])): \ W \subseteq [N], \ |W| = \mu \big\},
\] 
and 
\[
d^\perp_\mu \define \min \big\{\dist((\C(\bL[W]))^\perp): \ W \subseteq [N], \ |W| = \mu \big\}.
\]
Then it is straightforward to see that the results in Theorems~\ref{mainthm1} and~\ref{mainthm2}
still hold, with $d$ and $d^\perp$ being replaced by $d_\mu$ and $d^\perp_\mu$, respectively. 
\end{remark}

\subsection{Role of the Field Size}

The following example demonstrates that the use of index codes over larger fields might have a positive impact on the security level. More specifically, in that example, index codes over large fields significantly enhance the security, compared with index 
codes over small fields.

\medskip
\begin{example}
\label{ex2}
Suppose that the source $S$ has $n$ messages $x_1, x_2, \ldots, x_n$. Assume that there are $m < n$ receivers $R_1, R_2, \ldots, R_m$, and each receiver $R_i$ has the same set of side information, $\X_i = \{m + 1, m + 2, \ldots, n \}$. Assume also that each $R_i$ requires $x_i$, for $i \in [m]$. 

Any index code for this instance must have length at least $m$, since all the vectors $\bui+\be_i$, for some $\bui \lhd \X_i$, $i \in [m]$, are linearly independent over any field. 

If we employ an index code over $\ff_2$, by the fact that there are no nontrivial binary MDS codes, we deduce that the minimum distance $d$ of $\cl$ is at most $n - m$. Hence index codes over $\ff_2$ is not secure against some adversaries of strength $t = n - m - 1$. However, if we consider index codes over $\fq$ for sufficiently large $q$ ($q \geq n - 1$), there exists a $q$-ary MDS code $\mc$ with minimum distance exactly $n - m + 1$. By choosing $\bL$ so that $\cl = \mc$, the index code based on $\bL$ is secure against all adversaries of strength at most $t = n - m - 1$, which is strictly more secure than the those over $\ff_2$. To find such an $\bL$, let $\bM = (\bI_m|\bP)$ be a generator matrix in standard form of an $[n, m]_q$-MDS code, and then take $\bL = \bM^T$. Then $\bL[i] = \bui + \bei$, for some $\bui \lhd \X_i = \{m + 1, m + 2, \ldots, n \}$, $i \in [m]$. Therefore, by Corollary~\ref{coro:index_code}, $\bL$ corresponds to a linear index code for this instance.  

Note that if we employ an index code over $\ff_2$, then for large values of $n$ the minimum distance $d$ of $\cl$ is bounded from above
by the sphere-packing bound
\[
   d \le 2 n \cdot ( \en^{-1}( 1 - m/n) - \vep),
\]
where $\vep \rightarrow 0$ as $n \rightarrow \infty$. 
There is a variety of stronger upper bounds on the minimum distance of binary codes, such as 
the Johnson bound, the Elias bound, and the McEliece-Rodemich-Rumsey-Welch bound (see~\cite[Chapter 4.5]{roth} for more
details). These bounds provide even stronger bounds on the security of the binary scheme for this instance of the ICSI problem. 
By contrast, as shown above, by using a $q$-ary MDS code, the distance $d$ of $\cl$ can achieve 
the Singleton bound. 
It is well known that there is a significant gap between the Singleton bound and the sphere-packing bound 
(see~\cite[p.~111]{roth} for details). 
Therefore, for this instance of the ICSI problem, 
index codes over large fields provide significantly higher levels of security than  
those over binary field. 
\end{example}

\subsection{Application: Index Coding with Side and Restricted Information}
\label{sec:ICSRI}

In this section, we consider an extension of the ICSI problem, which we call the \emph{Index Coding with Side and Restricted Information (ICSRI)} problem. This problem arises in applications such as audio and video-on-demand. Consider a client who 
has subscribed for certain media content (audio or video programs, movies, newspapers, etc.) At the same time, this client has 
not subscribed to some other content. The content provider wants to restrict this client from obtaining a content which 
he is not eligible for, even though he might be able to obtain it ``for free'' from the transmissions provided by the server. 
As we show in sequel, the solution for the ICSRI problem is a straight-forward application of the results in Corollary~\ref{coro1}.

More formally, the arguments of an instance $(m,n,\X,\Z,f)$ of the ICSRI problem are similar to their counterparts for the ICSI problem. The new additional parameter, $\Z = (\Z_1,\Z_2,\ldots,\Z_m)$, represents the sets $\Z_i \subseteq [n]$ of message indices that the respective receivers $R_i$, $i \in [m]$, are not allowed to obtain. The goal is that at the end of the communication round, the receiver $R_i$ has the message $x_{f(i)}$ in its possession, for all $i \in [m]$, and it has no information about $x_j$ for all $j \in \Z_i$. The notion of a linear {\mic} is naturally extended to that of a linear $(m,n,\X,\Z,f)$-IC over $\fq$. 

Let 
\[
\FXZ \define \bigcup_{i = 1}^m \left\{ \bu + \be_j:\ \bu \lhd \X_i,\ j \in \Z_i \right\}. 
\]    
The following proposition provides a necessary and sufficient condition for a linear index code 
to be also a solution to an instance of the ICSRI problem.

\vskip 10pt 
\begin{proposition}
\label{pro:ICSRI}
The linear {\mic} based on $\bL$ is also a linear $(m,n,\X,\Z,f)$-IC if and only if 
$\cl \cap \FXZ = \varnothing$.
\end{proposition}
\begin{proof}
Let $S$ employ the {\mic} based on $\bL$. Then clearly $R_i$ can recover $x_{f(i)}$ for all $i \in [m]$. 
Due to Lemma~\ref{lem3}, for each $i \in [m]$ and $j \in \Z_i$, $R_i$ has no information about $x_j$ if and only if 
\[
\forall \bu \lhd \X_i: \ \bu + \be_j \notin \cl. 
\]   
Hence we complete the proof. 
\end{proof}
\vskip 10pt  

\begin{example}
Consider an instance $(m,n,\X,\Z,f)$ of the ICSRI problem where $m$, $n$, $\X$, and $f$ are defined as in Example~\ref{ex1}. 
Moreover, let $\Z = (\Z_1,\Z_2,\ldots,\Z_7)$, where $\Z_1 = \{2,3,4,5\}$, $\Z_2 = \{1,3,4,6\}$, $\Z_3 = \{1,2,4,7\}$, 
and $\Z_4 = \Z_5 = \Z_6 = \Z_7 = \varnothing$. Consider the index code based on $\bL$ 
constructed in Example~\ref{ex1}. It is straightforward to verify that $\cl \cap \FXZ = \varnothing$. Therefore, 
by Proposition~\ref{pro:ICSRI}, this index code also provides a solution to this instance of the ICSRI problem.   
\end{example} 
\vskip 10pt 

Let 
\[
\begin{split}
\kp^*_q &= \kp^*_q(m,n,\X,\Z,f) \\
&\define \min\{\rank(\{\bui + \be_{f(i)}\}_{i \in [m]})\},
\end{split} 
\]
where the minimum is taken over all choices of $\bui \lhd \X_i$, $i \in [m]$, which satisfy
\begin{equation}
\label{1:E100}
\spn\left(\{\bui + \be_{f(i)}\}_{i \in [m]}\right) \cap \FXZ = \varnothing. 
\end{equation} 
Let $\kp^*_q = +\infty$ if there are no choices of $\bui$'s, $i \in [m]$, which satisfy (\ref{1:E100}).  
The following proposition follows immediately. 

\vskip 10pt 
\begin{proposition}
The length of an optimal linear $(m,n,\X,\Z,f)$-IC over $\fq$ is $\kp_q^*$. If $\kp^*_q = +\infty$ then 
there exist no linear $(m,n,\X,\Z,f)$-ICs over $\fq$. 
\end{proposition}

\section{Strongly Secure Index Codes with Side Information}
\label{sec:strong_security}

In this section, we consider a different model of adversary. 
Similarly to its counterpart in Section~\ref{sec:block_security}, the adversary $A$ in this section owns some prior side information. Additionally, $A$ can listen to $\mu \le N$ transmissions of $S$. It can also corrupt some transmissions of $S$, received by any of $R_i$, $i \in [m]$. 

We start the analysis with some basic definitions of error-correcting index codes.
This type of index codes was studied very recently by the authors of this paper in~\cite{Dau-ECIC}. 
We repeat some basic results for the sake of completeness. 

\subsection{Error-Correcting Index Codes}
\label{subsec:error-correcting}

Assume that some of the symbols received by $R_i$, $i \in [m]$, are in error.  
Consider an ICSI instance $(m,n,\X,f)$, and assume that $S$ broadcasts 
a vector $\fkE(\bx) \in \fq^N$. 
Let $\bxii \in \fq^N$ be the error affecting the information received by $R_i$, $i \in [m]$. 
Then $R_i$ actually receives the vector
\[
\byi = \fkE(\bx) + \bxii \in \fq^N \; ,
\]
instead of $\fkE(\bx)$. The following definition is a generalization of Definition~\ref{def:IC}. 

\medskip
\begin{definition} 
\label{def:ECIC}
A \emph{$\delta$-error-correcting index code} over $\fq$ for an instance $(m,n,\X,f)$
of the ICSI problem, referred to as a $\delta$-error-correcting {\mic}, is an encoding function
\begin{eqnarray*}
\fkE & : & \fq^n \rightarrow \fq^N \; , 
\end{eqnarray*}
such that for each receiver $R_i$, $i \in [m]$, there exists a decoding function
\[
\fkD_i \: : \: \fq^N \times \fq^{|\X_i|} \rightarrow \fq \; , \\
\]
satisfying
\[
\forall \bx, \bxii \in \fq^n, \; \weight(\bxii) \le \delta \; : \; \fkD_i(\fkE(\bx) + \bxii, \bx_{\X_i}) = x_{f(i)}\; .
\]
\end{definition}
The definitions of the length, of a linear index code, and of the matrix corresponding 
to an index code are naturally extended to $\delta$-error-correcting index codes. 
\vskip 10pt 

%
We define the following sets
\[
\begin{split} 
\I(& q, m, n,\X,f)\\
& \define \{\bz \in \fqn: \exists i \in [m] \text{ such that } \bz_{\X_{i}} = \bO \text{ and } z_{f(i)} \neq 0\}.
\end{split} 
\]
For each $i \in [m]$ we also define 
\[
\Y_i \define [n] \backslash \Big( \{f(i)\} \cup \X_i \Big).
\] 
Then the collection of supports of all vectors in $\I(q, m, n, \X, f)$ is precisely
\begin{equation}
\label{Jdef} 
\J(m,n,\X,f) \define \bigcup_{i \in [m]} \Big \{ \{f(i)\} \cup Y_i: Y_i \subseteq \Y_i \Big\}.
\end{equation} 
The necessary and sufficient condition for a matrix $\bL$ to correspond to a linear 
$\delta$-error-correcting index code is given in the following lemma. 

\vskip 10pt
\begin{lemma}
\label{lem:decodability}
The matrix $\bL$ corresponds to a linear {\dd} if and only if 
\begin{equation} 
\label{2:E9}
\weight(\bz \bL) \geq 2\delta+1 \text{ for all } \bz \in \I(q,m,n,\X,f). 
\end{equation} 
Equivalently, $\bL$ corresponds to a linear {\dd} if and only if 
\[
\weight \left( \sum_{i \in K}z_i\bL_i \right) \geq 2\delta + 1,
\]
for all $K \in \JX$ and for all choices of nonzero $z_i \in \fq$, $i \in K$. 
\end{lemma}

\begin{proof}
For each $\bx \in \fq^n$, we define
\[
B(\bx,\delta) = \{\bch \in \fqn:\ \bch = \bx \bL + \bxi, \ \weight(\bxi) \leq \delta,\ \bxi \in \fqn\},
\] 
the set of all vectors resulting from at most $\delta$ errors in the transmitted vector 
associated with the information vector $\bx$. 
Then the receiver $R_i$ can recover $x_{f(i)}$ correctly
if and only if 
\[
B(\bx,\delta) \cap B(\bx',\delta) = \varnothing,
\]
for every pair $\bx,\bx' \in \fqn$ satisfying:
\[
\bx_{\X_i} = \bx'_{\X_i} \text{ and } x_{f(i)} \neq x'_{f(i)}. 
\]
Therefore, $\bL$ correspond to a linear {\dd} if and only if the following condition is satisfied: 
for all $i \in [m]$ and for all $\bx,\bx' \in \fqn$ such that
$\bx_{\X_i} = \bx'_{\X_i}$ and $x_{f(i)} \neq x'_{f(i)}$, it holds 
\begin{multline}
\forall \bxi, \bxi' \in \fq^N, \; \weight(\bxi) \le \delta, \; \weight(\bxi') \le \delta \; : \\
 \bx \bL + \bxi \neq \bx' \bL + \bxi' \; . 
\label{eq:unique-decode}
\end{multline}
Denote $\bz = \bx' - \bx$. Then, the condition in~(\ref{eq:unique-decode}) can be reformulated as follows: 
for all $i \in [n]$ and for all $\bz \in \fqn$ such that $\bz_{\X_i} = \bO$ and $\bz_{f(i)} \neq 0$, it holds
\begin{multline}
\forall \bxi, \bxi' \in \fq^N, \; \weight(\bxi) \le \delta, \; \weight(\bxi') \le \delta \; : \;
\bz \bL \neq \bxi -\bxi' \; . 
\label{eq:unique-decode-z}
\end{multline}
The equivalent condition is that for all $\bz \in \IX$, 
\[
\weight(\bz \bL) \ge 2 \delta + 1 \; . 
\]
Since for $\bz \in \IX$ we have
\[
\bz \bL = \sum_{i \in \supp(\bz)} z_i \bL_i,
\]
the condition (\ref{2:E9}) can be restated as
\[
\weight \left( \sum_{i \in K}z_i\bL_i \right) \geq 2 \delta + 1,
\]
for all $K \in \JX$ and for all choices of nonzero $z_i \in \fq$, $i \in K$. 
\end{proof}
\vskip 10pt

The next corollary follows directly from Lemma~\ref{lem:decodability} by considering an error-free
setup, i.e. $\delta = 0$. It is easy to verify that the conditions stated in this corollary and in Corollary~\ref{coro:index_code}
are equivalent, as expected. 
 
\vskip 10pt
\begin{corollary}
\label{coro:ic_decodability}
The matrix $\bL$ corresponds to an {\mic} if and only if 
\[
\weight \left( \sum_{i \in K}z_i\bL_i \right) \geq 1,
\]
for all $K \in \JX$ and for all choices of nonzero $z_i \in \fq$, $i \in K$. 
\end{corollary}

\subsection{A Lower Bound on the Length}

We start this section with a generalization of the definition of index codes to \emph{randomized index codes}. 
Consider $\eta \in \bbN$ random variables $G_1,G_2,\ldots,G_\eta$, which are distributed independently and uniformly over $\fq$. Let $\bG = (G_1,G_2,\ldots,G_\eta)$ and let $\bg= (g_1,g_2,\ldots,g_\eta)$ be a realization of $\bG$. 
\medskip

\begin{definition} 
An \emph{\rmic} for an instance $(m, n,\X,f)$ 
is an encoding function
\begin{eqnarray*}
\fkE & : & \fq^n \times \fq^\eta \rightarrow \fq^N \; , 
\end{eqnarray*}
such that for each receiver $R_i$, $i \in [m]$, there exists a decoding function
\[
\fkD_i \: : \: \fq^N \times \fq^{|\X_i|} \rightarrow \fq \; , \\
\]
satisfying
\[
\forall \bx \in \fq^n \; : \; \fkD_i(\fkE(\bx, \bg), \bx_{\X_i}) = x_{f(i)} \; ,
\]
for any $\bg \in \fq^\eta$, which is a realization of the random vector $\bG$.
\label{def:RIC}
\end{definition}
\vskip 10pt 

The definition of a $\delta$-error-correcting index code can be naturally 
extended to that of a $\delta$-error-correcting randomized index code. 
We simply replace $\fkE: \fq^n \ra \fq^N$ by $\fkE: \fq^n \times \fq^\eta \ra \fq^N$,
and $\fkE(\bx)$ by $\fkE(\bx,\bg)$ in Definition~\ref{def:ECIC}. 

An $\eta$-randomized index code is linear over $\fq$ if it has a linear encoding function $\fkE$, 
\[
\fkE(\bx,\bg)=(\bx \; | \; \bg)\bL \; , 
\]
where $\bL$ is an $(n + \eta) \times N$ matrix over $\fq$.
In the sequel we assume that any message $x_i$, $i \in [n]$ is requested by at least one receiver.
Observe that by simply treating $x_1,x_2,\ldots,x_n,g_1,g_2,\ldots,g_\eta$ as messages, 
the results from previous sections still apply to linear randomized index codes.    
 
\vskip 10pt 
\begin{definition}
The linear {\rmic} based on $\bL$ is said to be \emph{$(\mu,t,\delta)$-strongly secure} if it has the following two properties: 
\begin{enumerate}
  \item This code is $\delta$-error-correcting. In other words, upon receiving $(\bx|\bg)\bL$ with at most $\delta$ coordinates in error, the receiver $R_i$ can still recover $x_{f(i)}$, for all $i \in [m]$. 
	\item This code is $(\mu,t)$-strongly secure. In other words, an adversary $A$ who possesses $\bx_{\X_A}$, for $\X_A \subseteq [n]$, $|\X_A| = t$, and listens to at most $\mu$ transmissions, $\mu \le N$,
gains no information about other messages. Equivalently, 
\[
\en(\bX_{\Xah} \; | \; (\bX|\bG) \bL[W], \bX_{\X_A}) = \en(\bX_{\Xah}),
\]	
for any $W \subseteq [N]$, $|W| \le \mu$. 
\end{enumerate}
\label{def:strong-secure}
\end{definition} 
\vskip 10pt 

\begin{remark}
\mbox{}
\begin{enumerate}
	\item If $\mu = t = \eta = 0$, then a $(\mu,t,\delta)$-strongly secure {\rmic}
	is simply a $\delta$-error-correcting $(m,n,\X,f)$-IC over $\fq$. 
 \item If $\delta = 0$, the index code is strongly secure, but has no error-correcting capability. In that case, we simply say that the code is ``$(\mu,t)$-strongly secure" instead of ``$(\mu,t,0)$-strongly secure".  
 \item A simple concatenation of an error-correcting index coding scheme and a secure index coding scheme 
 may not necessarily yield
 a $(\mu,t,\delta)$-strongly secure {\rmic}.
\end{enumerate}
\end{remark}
\vskip 10pt 

In the lemma below, we assume that each message is requested by at least one receiver. Otherwise, 
that ``useless'' message can be discarded without affecting the model.

\vskip 10pt 
\begin{lemma} 
\label{lem:randomness}
If $\bL$ corresponds to a $(\mu,t)$-strongly secure linear {\rmic}, then $\eta \geq \mu$. 
\end{lemma} 

\begin{proof} 
We prove this lemma by contradiction. Suppose that $\bL$ corresponds to a $(\mu,t,\delta)$-strongly secure 
{\rmic}, and that $\eta < \mu$. Let $E = \{n+1,n+2,\ldots,n+\eta\}$. 

For $W \subseteq [N]$ let $\clw$ be the space spanned by columns of $\bL$ indexed by elements of $W$.
Then, for all $W \subseteq [N]$ with $|W| \leq \mu$, it holds that
\[
\en(\bX_{\Xah}|(\bX|\bG)\bL[W], \bX_{\Xa}) = \en(\bX_{\Xah}),
\]
i.e. an adversary who owns $\bx_\Xa$ gains no information about $\bx_\Xah$ after eavesdropping 
the transmissions corresponding to the set of indices $W$. 
From Lemma~\ref{lem3} with $\cl$ being replaced by $\clw$, we conclude that $\clw$ does not contain 
a vector $\bc$ which satisfies $\bc_\Xah \neq \bO$ and $\bc_E = \bO$. 
In the sequel, we refer to this property of $\clw$ as Property~A.

Let $\bL'= (\bL_{\Xah \cup E})^T$ be the matrix obtained from $\bL$ by first deleting rows of $\bL$ indexed by $\Xa$, and then
taking its transpose. We show that $\rank(\bL') \leq \mu - 1$. Indeed, take any $\mu$ rows of $\bL'$, 
denote them $\bL'_{j_1},\ldots, \bL'_{j_\mu}$.  
Let $\bL''$ be the submatrix of $\bL'$ formed by the last
$\eta$ columns. Since $\eta < \mu$, the $\mu$ rows $\bL''_{j_1},\ldots,\bL''_{j_\mu}$ are linearly dependent. 
Hence, there exist $\al_1,\al_2,\ldots,\al_\mu$, not all zeros, such that
\[
\sum_{\ell =1}^\mu \al_\ell \bL''_{j_\ell} = \bO \; .
\]
This implies
\[
\sum_{\ell =1}^\mu \al_\ell \bL'_{j_\ell} = \bO \; ,
\]
due to Property A. Thus, $\rank(\bL') \leq \mu - 1$. 

Now let $\vr \define \rank(\bL') < \mu$, and let 
\[
\{\bL'_{j_1},\bL'_{j_2},\ldots, \bL'_{j_\vr} \}
\]
be a basis of the space spanned by the rows of $\bL'$. 
Suppose that the receiver $R_i$ requests $x_{f(i)}$ where $f(i) \in \Xah$.
\begin{itemize}
\item
On the one hand, by Corollary~\ref{coro:index_code}, $\cl$ contains a vector $\bc = \bui + \be_{f(i)}$ 
where $\bui \lhd \X_i$. 
Therefore, $\bc_E = \bO$ and $\bc_\Xah \neq \bO$. 
\item
On the other hand, there exist $\bt_1,\bt_2,\ldots,\bt_\vr$ such that
\[
(\bc_\Xah|\bc_E) = \sum_{\ell = 1}^\vr \bt_\ell \bL'_{j_\ell}.
\] 
Since $\vr < \mu$ and $\bc_E = \bO$, by Property A we have $\bc_\Xah = \bO$. 
\end{itemize}
We obtain a contradiction. 
\end{proof} 
\vskip 10pt 

\begin{remark}
From Lemma~\ref{lem:randomness}, a $(\mu,t,\delta)$-strongly secure linear randomized index code requires at least $\mu$ 
random symbols. We show in Section~\ref{sec:optimal_code} that there exists such a code that uses precisely $\mu$ random symbols. 
\end{remark}

\vskip 10pt 
\begin{lemma}
\label{lem:randomness_recovery}
Suppose that $\bL$ corresponds to a linear $\mu$-randomized {\mic}. If this randomized index code is $(\mu, t)$-strongly secure, then 
for all $i \in [\mu]$, there exists a vector $\bvi \in \fq^{n + \mu}$ satisfying
\begin{enumerate}
	\item $\bvi \lhd [n]$;
	\item $\bvi + \be_{n + i} \in \cl$. 
\end{enumerate}
\end{lemma}

\begin{proof}
Assume, by contradiction, that for some $i \in [\mu]$, we have $\bvi + \be_{n + i} \notin \cl$ for all
$\bvi \lhd [n]$.
Consider a virtual receiver, which has a side information set $\{x_j\}_{j \in [n]}$, 
and requests the symbol $g_i$.  
By Corollary~\ref{coro1}, this virtual receiver has no information about $g_i$ after
listening to all transmissions. In other words, we have
\begin{equation} 
\label{E20}
\en(G_i| (\bX | \bG) \bL, \bX) = \en(G_i) \; , 
\end{equation} 
and, in particular, for a smaller set of side information, 
\begin{equation} 
\label{E20-1}
\en(G_i| (\bX | \bG) \bL, \bX_\Xa) = \en(G_i) \; . 
\end{equation} 
We recall Definition~\ref{def:strong-secure}: for every $\mu$-subset $W \subseteq [N]$ and every $t$-subset $\Xa \subseteq [n]$, we have
\begin{equation} 
\label{E21}
\en(\bX_\Xah| (\bX|\bG) \bL[W], \bX_\Xa) = \en(\bX_\Xah) \; .
\end{equation} 
In the sequel we show that if the value of $G_i$ is known to the adversary, 
this randomized index code is still $(\mu,t)$-strongly secure. 
In other words, we aim to show that
\begin{equation} 
\label{E22}
\en(\bX_\Xah| (\bX | \bG) \bL[W], \bX_\Xa, G_i) = \en(\bX_\Xah) \; ,
\end{equation} 
for every $\mu$-subset $W \subseteq [N]$ and every $t$-subset $\Xa \subseteq [n]$. 
Indeed, the left-hand side of (\ref{E22}) is equal to 
\[
\en(\bX_\Xah| (\bX | \bG) \bL[W], \bX_\Xa) - \mi(\bX_\Xah; G_i| (\bX | \bG) \bL[W], \bX_\Xa),
\]
which is 
\[
\en(\bX_\Xah) - \mi(\bX_\Xah; G_i| (\bX | \bG) \bL[W], \bX_\Xa)
\]
due to (\ref{E21}). Hence, it suffices to show that 
\[
\mi(\bX_\Xah; G_i \; | \; (\bX | \bG) \bL[W], \bX_\Xa) = 0 \; . 
\]
We have
\[
\begin{split}
\mi(\bX_\Xah; G_i| (\bX | \bG) & \bL[W], \bX_\Xa)\\ 
&= \; \en(G_i| (\bX | \bG) \bL[W], \bX_\Xa)\\
& \; \quad - \en(G_i|(\bX | \bG) \bL[W], \bX_\Xa, \bX_\Xah)\\
&= \; \en(G_i| (\bX | \bG) \bL[W], \bX_\Xa)\\
& \; \quad - \en(G_i|(\bX | \bG) \bL[W], \bX)\\
&= \; \en(G_i) - \en(G_i)\\
&= \; 0\ \; ,
\end{split}
\]
where the third transition is due to~(\ref{E20}) and~(\ref{E20-1}).

To this end, we have shown that the randomized index code is still $(\mu,t)$-strongly secure if the adversary knows the 
realized value of $G_i$. Equivalently, discarding the random variable $G_i$ from the scheme does not affect its strong security. However, 
this contradicts Lemma~\ref{lem:randomness}, since the resulting code has less than $\mu$ random
symbols.  
\end{proof} 
\vskip 10pt 

The following theorem proves a lower bound on the length of a $(\mu,t)$-strongly secure linear randomized index code. 

\vskip 10pt 
\begin{theorem}
\label{thm:lowerbound}
The length of a $(\mu,t)$-strongly secure linear $\eta$-randomized {\mic} is at least $\kpq + \mu$. 
\end{theorem} 

\begin{proof}
Suppose the linear randomized index code is based on $\bL$. 
We divide the proof into several cases. 
\begin{description}
\item[Case 1:] 
$\;\; \eta = \mu$. Then, by Corollary~\ref{coro:index_code} and Lemma \ref{lem:randomness_recovery}, the subspace $\cl$ must contain:  
\begin{itemize}
\item
the vectors $\bui + \be_{f(i)}$ for some $\bui \lhd \X_i$, for all $i \in [m]$;
\item
the vectors $\bvi + \be_{n + i}$, for some $\bvi \lhd [n]$, for all $i \in [\mu]$. 
\end{itemize}
Due to linear independence of these vectors and to the definition of $\kpq$, the length of the code is at least
\[
\begin{split}
\dim(\cl) & \geq \rank(\{\bui + \be_{f(i)}\}_{i \in [m]})\\
& \quad + \rank(\{\bvi + \be_{n + i}\}_{i \in [\mu]})\\
& \geq \kpq + \mu \; . 
\end{split}
\] 

\item[Case 2:] $\;\; \eta > \mu$, and for all $i \in [\eta]$ there exists some vector $\bvi \lhd [n]$ such that $\bu^{(i)} + \be_{n + i} \in \cl$.

In this case, similarly to Case 1, we have
	\[
	\dim(\cl) \geq \kpq + \eta > \kpq + \mu.
	\]
	Therefore, $\bL$ has at least $\kpq + \mu$ columns.
	 
\item[Case 3:] $\;\; \eta > \mu$, and for some $i \in [\eta]$, $\bvi + \be_{n + i} \notin \cl$ for all $\bvi \lhd [n]$. By following exactly the same argument as in the proof of Lemma~\ref{lem:randomness_recovery}, we deduce that discarding $G_i$ does not affect the strong security of the randomized index code. By doing so, we obtain a new randomized $(\mu,t)$-strongly secure index code, which has $\eta - 1$ random variables. This code is based on $\bL'$, which is obtained from $\bL$ by deleting its $(n+i)$-th row. 

The above argument can be applied until either the number of random variables decreses to $\mu$, 
or the code in consideration satisfies the condition of Case 2. 
In both cases, the resulting randomized index code has length at least $\kpq + \mu$. As the length of the code do not change during the process, we conclude that the length of the original code is at least $\kpq + \mu$.    
\end{description}
\end{proof} 
\vskip 10pt 

The next theorem establishes a lower bound on the length of a $(\mu,t,\delta)$-strongly secure linear randomized index code. 

\vskip 10pt
\begin{theorem}
\label{thm:lowerbound2}
The length of a $(\mu,t, \delta)$-strongly secure linear {\rmic} is at least $\kpq + \mu + 2\delta$.
\end{theorem}

\begin{proof}
Let $\bL$ correspond to a $(\mu,t, \delta)$-strongly secure {\rmic}.
Let $\bL'$ be the matrix obtained by deleting any $2\delta$ columns of $\bL$. 
Since $\bL$ corresponds to a $\delta$-error-correcting index code, 
by Lemma~\ref{lem:decodability} it satisfies
\[
\weight \left( \sum_{i \in K} z_i \bL_i \right) \geq 2\delta + 1 \; ,
\]
for all $K \in \JX$ and all choices of nonzero $z_i \in \fq$, $i \in K$. We 
obtain that the rows of $\bL'$ satisfy
\[
\weight \left( \sum_{i \in K} z_i \bL'_i \right) \geq 1 \; .
\] 
By Corollary~\ref{coro:ic_decodability}, $\bL'$ corresponds to an {\rmic}. 
Since all entries of $\bL'$ are contained in $\bL$, 
we deduce that $\bL'$ corresponds to a $(\mu,t)$-strongly secure {\rmic}.
Therefore, by Theorem~\ref {thm:lowerbound}, $\bL'$ has at least $\kpq + \mu$ columns. 
Therefore, $\bL$ has at least $\kpq + \mu + 2\delta$ columns.  
\end{proof} 
\vskip 10pt 

\subsection{A Construction of Optimal Strongly Secure Index Codes}
\label{sec:optimal_code}
In this section, we present a construction of an optimal $(\mu, t, \delta)$-strongly secure $\mu$-randomized
linear $(m,n,\X,f)$-IC over $\fq$, which has length attaining the lower bound established in Theorem~\ref{thm:lowerbound2}.
It requires $q$ to be at least $\kpq + \mu + 2\delta + 1$. 
The proposed construction is based on the coset coding technique, originally introduced by Ozarow and Wyner~\cite{OzarowWyner1984}. This technique has been adopted in a variety of network coding applications, such as~\cite{CaiYeung2002, Feldman2004, Rouayheb_Soljanin2007, Zhuang2010, Silva_Kschischang2010}.

{\bf Construction A:} Let $\bLo$ correspond to a linear {\mic} of optimal length $\kpq$. Let $\bM$ be a generator matrix
of an $[N = \kpq + \mu + 2\delta, \kpq + \mu, 2 \delta + 1]_q$ MDS code, so that the last $\mu$ rows of $\bM$ form a generator matrix
of another MDS code. For instance, take
\[
\bM = 
\left( 
\begin{array}{cccc}
1 & 1 & \cdots & 1 \\
\al_1 & \al_2 & \cdots & \al_N \\
\vdots & \vdots & \ddots & \vdots \\
\vspace{1ex}
\al_1^{{\kpq} - 1} & \al_2^{{\kpq} - 1} & \cdots & \al_N^{{\kpq} - 1} \\
\hline 
\vspace{-2ex} &&& \\
\al_1^{\kpq} & \al_2^{\kpq} & \cdots & \al_N^{\kpq} \\
\al_1^{{\kpq} + 1} & \al_2^{{\kpq} + 1} & \cdots & \al_N^{{\kpq} + 1} \\
\vdots & \vdots & \ddots & \vdots \\
\al_1^{\kpq + \mu - 1} & \al_2^{\kpq + \mu - 1} & \cdots & \al_N^{\kpq + \mu - 1} \\
\end{array} \right) ,
\]  
where $\al_1,\al_2,\ldots,\al_N$ are pairwise distinct nonzero elements in $\fq$. 
Let $\bP$ be the submatrix of $\bM$ formed by the first $\kpq$ rows, and $\bQ$
the submatrix formed by the last $\mu$ rows of $\bM$. 
Take 
\[
\bL=\left(
\begin{array}{c}
\bLo \bP \\ \hline
\bQ \\ 
\end{array}\right).
\]

\vskip 10pt 
\begin{lemma}
\label{lem:error_correction}
The matrix $\bL$ in Construction A corresponds to a $\delta$-error-correcting $\mu$-randomized $(m,n,\X,f)$-IC over $\fq$. 
\end{lemma}

\begin{proof}
Recall that $\bg \in \fq^\mu$ is a random vector. The encoding function $\fkE$ has a form
\[
\fkE(\bx,\bg) = (\bx|\bg) \bL = \bx \bLo \bP + \bg \bQ = (\bx \bLo|\bg) \bM \; . 
\] 
Since $\bM$ is a generator matrix of a $\delta$-error-correcting code,  
each receiver $R_i$, $i \in [m]$, is able to recover $(\bx \bLo|\bg)$
if the number of errors in $\fkE(\bx,\bg)$ is less than or equal to $\delta$. 
Therefore, each receiver $R_i$ can recover $\bx \bLo$, and hence, it can also 
recover $x_{f(i)}$, $i \in [m]$, as $\bLo$ corresponds to a linear {\mic}. 
\end{proof} 
\vskip 10pt 

\begin{lemma}
\label{lem:strong_security}
The matrix $\bL$ in Construction A corresponds to a $(\mu,t)$-strongly secure $\mu$-randomized $(m,n,\X,f)$-IC over $\fq$. 
\end{lemma}
\begin{proof} 
Suppose that the adversary $A$ possess a message vector $\bx_\Xa$, $|\bx_\Xa| = t$.
Additionally, $A$ can eavesdrop $\mu$ transmissions, i.e. it has a knowledge of
$\bb \define (\bx|\bg) \bL[W]$, for some $W \subseteq [N]$, $|W| = \mu$. Below, we show that
the entropy of $\bX_\Xah$ is not changed given the knowledge of $(\bX|\bG)\bL[W]$ and of $\bx_\Xa$. 
It suffices to show that for all $\ba \in \fq^{n-t}$:
\begin{equation} 
\label{E23}
\text{Pr}(\bX_\Xah = \ba \; | \; (\bX|\bG)\bL[W] = \bb, \; \bX_\Xa = \bx_\Xa) = \dfrac{1}{q^{n - t}} \; .
\end{equation} 
The left-hand side of (\ref{E23}) can be re-written as
\begin{equation} 
\label{E24}
\dfrac{\text{Pr}(\bX_\Xah = \ba, (\bX|\bG)\bL[W] = \bb \; | \; \bX_\Xa = \bx_\Xa)}{\text{Pr}((\bX|\bG)\bL[W] = \bb 
\; | \; \bX_\Xa = \bx_\Xa)}. 
\end{equation} 
The numerator in~(\ref{E24}) is given by
\begin{equation} 
\begin{split}
&\text{Pr}(\bX_\Xah = \ba, (\bX|\bG)\bL[W] = \bb \; | \; \bX_\Xa = \bx_\Xa) \\
& = \; \text{Pr}(\bX_\Xah = \ba \;|\; \bX_\Xa = \bx_\Xa)\\
& \qquad \times \text{Pr}((\bX|\bG)\bL[W] = \bb \;|\; \bx_\Xah = \ba, \bX_\Xa = \bx_\Xa)\\
& = \; \dfrac{1}{q^{n - t}} \text{Pr}((\bX|\bG)\bL[W] = \bb \;|\; \bX_\Xah = \ba, \bX_\Xa = \bx_\Xa)\\
& = \; \dfrac{1}{q^{n - t}} \dfrac{1}{q^\mu} \; = \; \dfrac{1}{q^{n - t + \mu}} .
\end{split} 
\label{E26}
\end{equation} 
The penultimate transition can be explained as follows. We have
\begin{equation} 
\label{E25}
\bb = (\bX|\bG)\bL[W] = \bX \bLo \bP[W] + \bG \bQ[W] \; . 
\end{equation} 
The matrix $\bQ[W]$ is invertible due to the fact that $\bQ$ is a generator matrix of 
an $[N,\mu]$-MDS code. Since $\bX$ is known, the system~(\ref{E25}) has a unique solution given by
\[
\bG = (\bb -  \bX \bLo \bP[W]) (\bQ[W])^{-1} \; . 
\]
Since $\bG$ is uniformly distributed over $\fq^\mu$, 
\[
\begin{split} 
&\text{Pr}((\bX|\bG)\bL[W] = \bb \; | \; \bX_\Xah = \ba, \bX_\Xa = \bx_\Xa)\\
& = \; \text{Pr}(\bG = (\bb -  \bX \bLo \bP[W]) (\bQ[W])^{-1})\\
& = \; \dfrac{1}{q^\mu} \; .
\end{split}
\]
Similarly to~(\ref{E26}), the denominator in~(\ref{E24}) is 
\begin{equation} 
\label{E27}
\begin{split} 
& \text{Pr}((\bX|\bG)\bL[W] = \bb \; | \;  \bX_\Xa = \bx_\Xa)\\
& = \; \sum_{\bc \in \fq^{n - t}} \text{Pr}(\bX_\Xah = \bc \; | \; \bX_\Xa = \bx_\Xa)\\
& \qquad \times \text{Pr}((\bX|\bG)\bL[W] = \bb \; | \; \bX_\Xah = \bc, \bX_\Xa = \bx_\Xa)\\
& = \; q^{n - t} \dfrac{1}{q^{n - t}} \dfrac{1}{q^\mu} \; = \; \dfrac{1}{q^\mu} \; .
\end{split} 
\end{equation} 
From (\ref{E24}), (\ref{E26}), and (\ref{E27}), we obtain (\ref{E23}), as claimed.  
\end{proof} 
\vskip 10pt 

From Theorem~\ref{thm:lowerbound2}, Lemma~\ref{lem:error_correction}, and Lemma~\ref{lem:strong_security}, we obtain the 
main result of this section. 

\vskip 10pt 
\begin{theorem} 
The length of an optimal $(\mu,t,\delta)$-strongly secure linear $\eta$-randomized $(m,n,\X,f)$-IC over $\fq$
($q \geq \kpq + \mu + 2\delta + 1$) is $\kpq + \mu + 2\delta$. Moreover, the code in Construction A achieves this
optimal length. 
\end{theorem}

\section{Conclusions and Open Questions}
\label{sec:conclusion}
In this paper, we initiate a study of the security aspects of linear index coding schemes. 
We introduce a notion of block security and establish two bounds on the security level of a linear index code
based on the matrix $\bL$. These analysis makes use of the minimum distance 
and the dual distance of $\cl$, the code spanned by the columns of $\bL$. 
While the dimension of this code corresponds to the number of transmissions in the scheme, 
the minimum distance characterizes its security strength.
 
Our second contribution is the analysis of the strong security of linear index codes. 
New bounds on the length of linear index codes, which are resistant to errors, eavesdropping, and information leaking, 
are established. Index codes that achieve these bounds are constructed. These new bounds cannot be deduced 
directly from the existing results in network coding literature. 

One important problem, which remains open, deals with a design of an optimal secure index coding scheme. 
This problem can be formulated as follows: given an instance of the ICSI problem, 
how to design $\bL$, such that $\cl$ has the largest possible minimum distance? 
More specifically, let us define the binary \emph{side information matrix} $\bA=(\sa_{i,j})_{i \in [n], \; j \in [n]}$ as in~\cite{Yossef}, namely 
\[
\sa_{i,j} = \left\{ \begin{array}{cl} 
1 & \mbox{if } j = i \mbox{ or } j \in \X_i \\
0 & \mbox{otherwise} 
\end{array} \right. \; . 
\]
The problem is equivalent to finding a way to turn certain off-diagonal $1$'s in $\bA$ into $0$'s, 
such that the rows of the resulting matrix generate an error-correcting code of the largest possible minimum distance. 
It is very likely that this task is a hard problem. For comparison,
even finding the minimum distance of an error-correcting code given by its generating matrix 
is known to be NP-hard~\cite{Vardy1997}.

\section{Acknowledgements}

The authors would like to thank Fr\'{e}d\'{e}rique Oggier for helpful discussions. This work is 
supported by the National Research Foundation of Singapore (Research Grant
NRF-CRP2-2007-03).

\bibliographystyle{IEEEtran}
\bibliography{OnSecureICwithSI}

\end{document}